\documentclass[12pt]{amsart}
\usepackage{amsmath,amssymb,graphicx,a4wide}

 \newtheorem{thm}{Theorem}[section]

 \newtheorem{prop}[thm]{Proposition}
 \theoremstyle{definition}
 \newtheorem{defn}[thm]{Definition}
 \theoremstyle{remark}
 \newtheorem{rem}[thm]{Remark}
 
 \numberwithin{equation}{section}

\newcommand{\ts}{\hspace{0.5pt}}
\newcommand{\nts}{\hspace{-0.5pt}}
\newcommand{\dd}{\,\mathrm{d}}

\begin{document}

\title[Continuous Diffraction]
 {Non-Periodic Systems with \\[2mm] Continuous Diffraction Measures}

\author[Baake]{Michael Baake}

\address{%
Fakult\"{a}t f\"{u}r Mathematik,
Universit\"{a}t Bielefeld,
Postfach 100131,\newline
\indent 33501 Bielefeld,
Germany}
\email{mbaake@math.uni-bielefeld.de}

\author[Birkner]{Matthias Birkner}
\address{%
Institut f\"{u}r Mathematik,
Johannes-Gutenberg-Universit\"{a}t,
Staudingerweg 9, \newline
\indent
55099 Mainz,
Germany}
\email{birkner@mathematik.uni-mainz.de}

\author[Grimm]{Uwe Grimm}
\address{%
Department of Mathematics and Statistics,
The Open University,
Walton Hall, \newline
\indent Milton Keynes MK7 6AA,
United Kingdom}
\email{uwe.grimm@open.ac.uk}


\begin{abstract}
  The present state of mathematical diffraction theory for systems
  with continuous spectral components is reviewed and extended. We
  begin with a discussion of various characteristic examples with
  singular or absolutely continuous diffraction, and then continue
  with a more general exposition of a systematic approach via
  stationary stochastic point processes. Here, the intensity measure
  of the Palm measure takes the role of the autocorrelation measure in
  the traditional approach.  We furthermore introduce a `Palm-type'
  measure for general complex-valued random measures that are
  stationary and ergodic, and relate its intensity measure to the
  autocorrelation measure.
\end{abstract}

\maketitle

\section{Introduction}

The (mathematical or kinematic) diffraction theory of systems in
Euclidean space with pure point spectrum is rather well
understood. Ultimately, this is due to the availability of Poisson's
summation formula and its generalisations to the setting of measures
(or to tempered distributions); see \cite[Sec.~9.2]{BBG-TAO} for a
systematic exposition. Beyond results on the spectral nature, this
often also provides explicit formulas for the diffraction measure,
such as in the cases of lattice-periodic systems and model sets. For
these systems, there is also a well-understood connection with the
Halmos--von Neumann theorem for the corresponding pure point dynamical
spectrum; see \cite{BBG-LMS,BBG-BL,BBG-LStru,BBG-BLM,BBG-BLvE} for
details as well as \cite{BBG-TAO} and references therein for general
background.

As soon as one enters the realm of systems with continuous diffraction
spectra (or at least with continuous spectral components), the
situation changes drastically. As in the case of Schr\"{o}dinger
operator spectra \cite{BBG-DEG}, much less is known about the plethora
of possibilities, and there rarely are explicit formulas for the
diffraction measures of specific examples. Until recently, explicit
results were restricted to simple systems of Bernoulli type (hence
with disorder that leads to independent random variables) or to some
paradigmatic examples in one dimension (and product systems built from
them).

There has now been some progress towards explicitly computable
examples in various directions \cite{BBG-Kai,BBG-BBM10,BBG-BKM}. In
particular, both for singular and for absolutely continuous cases,
constructive approaches have been more successful than previously
anticipated; compare \cite[Ch.~10]{BBG-TAO}. Consequently, there is
some hope that more systems can be understood in this way. This view
is also supported by the recent progress in the understanding of the
connection between the dynamical and the diffraction spectrum in this
more general situation; see \cite{BBG-BLvE} and references therein.
At the same time, such examples will improve our intuition about
systems with continuous diffraction. Below, this will be reflected by
several short sketches of characteristic examples (which are covered
in more detail in \cite{BBG-TAO}), before we embark on a more
systematic setting via general point process theory. Our focus is on
systems in $\mathbb{R}^{d}$, which is the primary situation to
understand, particularly from the applications point of
view. Extensions to more general locally compact Abelian groups are
possible, but will not be discussed here.

\section{Diffraction Measures --- a Brief Reminder}

Let $\omega$ be a locally finite (and possibly complex) measure on
$\mathbb{R}^{d}$, which we primarily view as a linear functional on
the space $C_{\mathsf{c}} (\mathbb{R}^{d})$ of continuous functions
with compact support on $\mathbb{R}^{d}$, together with some mild
extra conditions. In favourable cases, $\omega$ will be translation
bounded.  By the classic Riesz--Markov representation theorem, we may
identify the measures defined by this approach with regular Borel
measures; for a systematic exposition, we refer to
\cite{BBG-Hof,BBG-BL} as well as \cite[Chs.~8 and 9]{BBG-TAO} and
references therein. Particularly important examples comprise the
\emph{Dirac measure} $\delta_{x}$, defined by $\delta_{x}(g) := g(x)$
for $g\in C_{\mathsf{c}} (\mathbb{R}^{d})$, and measures of the form
\begin{equation}\label{BBG-eq:comb}
    \delta^{}_{\nts S} \, := \sum_{x\in S} \delta^{}_{x}\ts ,
\end{equation}
which are known as \emph{Dirac combs}, where $S\subset\mathbb{R}^{d}$
is uniformly discrete. More generally, we will also consider objects
of the form $\sum_{x\in S} w(x)\ts \delta_{x}$, which can be a measure
for a general countable set $S$, then under suitable conditions on the
weight function $w$. Such measures are referred to as \emph{weighted}
Dirac combs.

Recall from \cite{BBG-Hof} or \cite{BBG-TAO} that, if $\omega$ is a
measure on $\mathbb{R}^{d}$, the (inverted-conjugate) measure
$\widetilde{\omega}$ is defined by $\widetilde{\omega} (g) :=
\overline{\omega (\widetilde{g})}$ for $g\in C_{\mathsf{c}}
(\mathbb{R}^{d})$, where $\widetilde{g} (x) := \overline{g(-x)}$.
Given a measure $\omega$, consider its \emph{autocorrelation measure}
\begin{equation}\label{BBG-eq:autodef}
     \gamma \, = \, \gamma^{}_{\omega}  := \,
     \omega \circledast \widetilde{\omega} \ts ,
\end{equation}
where $\circledast$ denotes the volume averaged (or Eberlein)
convolution. The latter is defined by
\[
     \omega \circledast \widetilde{\omega} \, := 
     \lim_{r\to\infty} \, \frac{\omega_{r} * \widetilde{\omega_{r}}}
     {\mathrm{vol} (B_{r} (0))}
\]
with $B_{r} (0)$ the (open) ball of radius $r$ around the origin and
$\omega_{r} := \omega|^{}_{B_{r} (0)}$. At this stage, we assume the
existence of the limit. This will be discussed in more detail later.

If (as in many of our examples) $\omega$ is a Dirac comb with
lattice support, also $\gamma$ will be supported on the same
lattice (or a subset of it). Concretely, if 
\[
   \omega \, = \, w \, \delta^{}_{\mathbb{Z}} \, := 
   \sum_{n\in\mathbb{Z}} w(n) \, \delta_{n} \ts ,
\]
with a bounded weight function $w$ say, one finds $\gamma = \eta \ts
\delta^{}_{\mathbb{Z}}$ with the positive definite function $\eta \! :
\, \mathbb{Z} \longrightarrow \mathbb{C}$ being defined by
\begin{equation}\label{BBG-eq:eta}
  \begin{split}
   \eta (m) \, :=&  \lim_{N\to\infty} \frac{1}{2N+1}
   \sum_{n=-N}^{N} w (n) \, \overline{w (n-m)} \\
   \, =&  \lim_{N\to\infty} \frac{1}{2N+1}
   \sum_{n=-N}^{N}\overline{ w (n)} \, w (n+m) \ts ,
  \end{split}
\end{equation}
provided that all limits exist. In our exposition below, this
existence will follow by suitable applications of Birkhoff's ergodic
theorem, applied to the dynamical system of the shift action on the
orbit closure of the sequence $w$ or to a similar type of dynamical
system; compare \cite{BBG-BL} for a more general setting. One benefit
of this approach will emerge via the Herglotz--Bochner theorem
\cite{BBG-Katz}.

The autocorrelation measure $\gamma$ is \emph{positive definite} (or
of positive type) by construction, which means that
$\gamma(g*\widetilde{g})\ge 0$ for all $g\in C_{\mathsf{c}}
(\mathbb{R}^{d})$. It is thus Fourier transformable \cite{BBG-BF}, and
the Fourier transform $\widehat{\gamma}$ is a positive measure, called
the \emph{diffraction measure} of $\omega$; see \cite{BBG-Cow} for the
physics behind this notion, and \cite{BBG-Hof} as well as
\cite[Ch.~9]{BBG-TAO} for the mathematical theory. Within the
framework of kinematic diffraction, it describes the outcome of a
scattering experiment by quantifying how much intensity is scattered
into a given volume of $d$-space, and thus is the central object of
our interest. By the Lebesgue decomposition theorem, there is a unique
splitting
\[
   \widehat{\gamma} \, = \, \widehat{\gamma}_{\textsf{pp}} + 
   \widehat{\gamma}_{\textsf{sc}} + \widehat{\gamma}_{\textsf{ac}}
\]
of the diffraction measure into its \emph{pure point} part
$\widehat{\gamma}_{\textsf{pp}}$, its \emph{singular continuous} part
$\widehat{\gamma}_{\textsf{sc}}$ and its \emph{absolutely continuous}
part $ \widehat{\gamma}_{\textsf{ac}}$, with respect to Lebesgue
measure $\lambda$. The pure point part comprises the `Bragg peaks' (of
which there are at most countably many, so
$\widehat{\gamma}_{\textsf{pp}}$ is a sum over at most countably many
Dirac measures with positive weights), while the absolutely continuous
part corresponds to the diffuse `background' scattering which is given
by a locally integrable density relative to $\lambda$. The singular
continuous part is whatever remains --- if present, it is a measure
that gives no weight to single points, but is still concentrated to an
(uncountable) set of zero Lebesgue measure.

Measures $\omega$ which lead to a diffraction $\widehat{\gamma} \, =
\, \widehat{\gamma}_{\textsf{pp}}$ are called \emph{pure point
  diffractive}; examples include lattice-periodic measures and
measures based on model sets. These have been studied in detail in the
context of diffraction of crystals and quasicrystals; see
\cite{BBG-BG11} for a recent review and \cite[Chs.~8 and 9]{BBG-TAO}
for a systematic exposition. Here, we are concentrating on the other
two spectral components, which may also carry important information on
the (partial) order which is present in the underlying structure. Pure
point spectra are discussed in detail in
\cite{BBG-Crelle,BBG-BL,BBG-BL2,BBG-BLM,BBG-TAO,BBG-LR07,BBG-LM09,BBG-LMpre};
for related spectral problems in the context of Schr\"{o}dinger
operators, we refer to \cite{BBG-DEG}.

\section{Guiding Examples}

As mentioned above, the understanding of systems with continuous
diffraction components is less developed than that of pure point
diffractive ones. Still, a better intuition will emerge from a sample
of characteristic examples. It is the purpose of this section to
provide some of them, while we refer to the literature for further
ones \cite{BBG-BH,BBG-BM98,BBG-BBM10,BBG-BG11,BBG-BKM,BBG-TAO}.

\subsection{Thue--Morse Sequences}

Let us begin with a classic example from the theory of substitution
systems that leads to a singular continuous diffraction measure with
rather different features in comparison with the Cantor measure, the
latter being illustrated in Figure~\ref{BBG-fig:cantor}. Our example
has a long history, which can be extracted from
\cite{BBG-Wie27,BBG-Mah27,BBG-K,BBG-AS}. We confine ourselves to a
brief summary of the results, and refer to \cite{BBG-BG08,BBG-TAO} and
references therein for proofs and details.

\begin{figure}
\begin{center}
\includegraphics[width=0.9\textwidth]{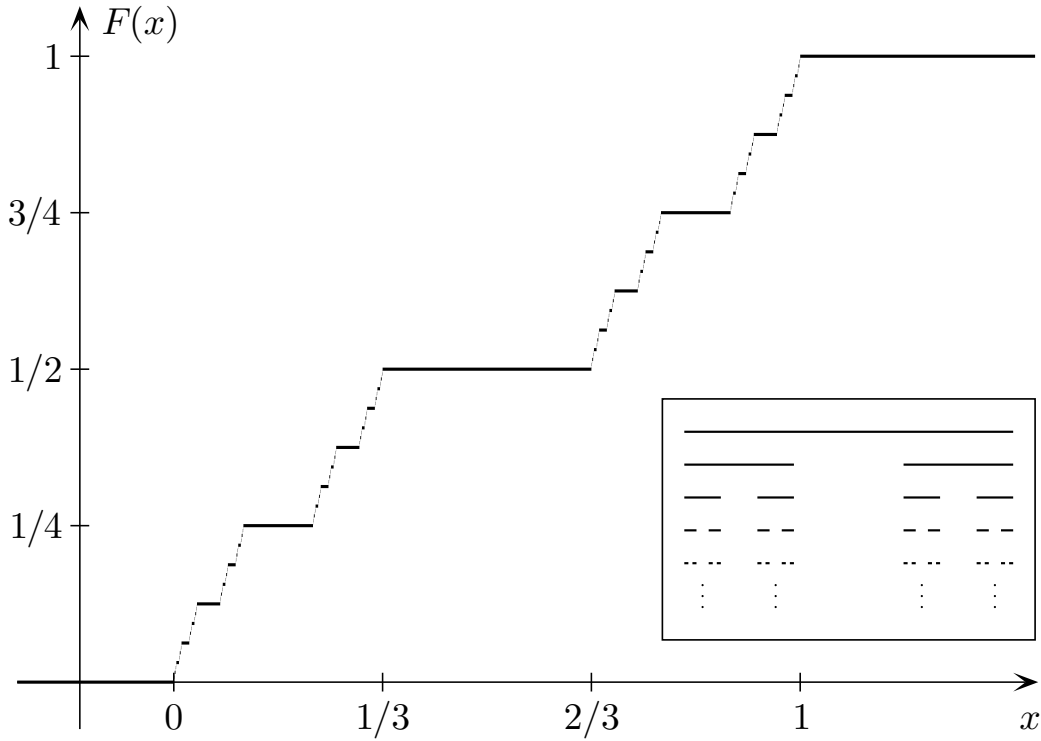}
\end{center}
\caption{\label{BBG-fig:cantor}The distribution function $F$ of the
  classic middle-thirds Cantor measure. The construction of the
  underlying Cantor set is sketched in the inset.}
\end{figure}

The classic \emph{Thue--Morse} (TM) sequence can be defined via the
one-sided fixed point $v=v^{}_{0}v^{}_{1}v^{}_{2}\ldots$ (with
$v^{}_{0}=1$) of the primitive substitution rule
\[
   \varrho\! : \, 
      \begin{array}{c} 1 \mapsto 1 \bar{1} \\ 
      \bar{1} \mapsto \bar{1} 1 \end{array}
\]
on the binary alphabet $\{1,\bar{1}\}$. The fixed point is the limit
(in the obvious product topology) of the (suitably embedded) iteration
sequence
\[
  1 \stackrel{\varrho}{\longmapsto} 
  1\bar{1}  \stackrel{\varrho}{\longmapsto} 1\bar{1}\bar{1}1
  \stackrel{\varrho}{\longmapsto}  1\bar{1}\bar{1}1 \bar{1}11\bar{1} 
   \stackrel{\varrho}{\longmapsto} \ldots 
  \longrightarrow v = \varrho (v) = v^{}_{0} v^{}_{1} v^{}_{2} 
  v^{}_{3} \ldots
\]
and has a number of distinctive properties \cite{BBG-AS,BBG-Q}, for
instance
\begin{itemize}
\item   $v^{}_{i} = (-1)^{\text{sum of the binary digits of $i$}}$
\item  $v^{}_{2i} \vphantom{\bar{v}} = v^{}_{i}$ and 
        $v^{}_{2i+1} = \overline{v^{}_{i}}$, for all $i\in\mathbb{N}^{}_{0}$;
\item   $v = v^{}_{0} v^{}_{2} v^{}_{4}\ldots$ and
        $\bar{v} = v^{}_{1} v^{}_{3} v^{}_{5}\ldots$
\item   $v$ is (strongly) cube-free (and hence non-periodic).
\end{itemize}
Here, we define $\bar{\bar{1}}=1$ and identify $\bar{1}$ with $-1$,
also for the later calculations with Dirac combs.  A two-sided
sequence $w$ can be defined by
\[  
  w (i) \,  = \, \begin{cases}
         v_{i} , & \text{for $i\ge 0$,} \\ 
         v_{-i-1}, & \text{for $i<0$,} \end{cases}
\]
which is a fixed point of $\varrho^{2}$, because the seed
$w^{}_{-1}|w^{}_{0} = 1|1$ is a legal word (it occurs in
$\varrho^{3}(1)$) and $w=\varrho^{2}(w)$. The (discrete) hull
$\mathbb{X}=\mathbb{X}^{}_{\mathrm{TM}}$ of the TM substitution is the
closure of the orbit of $w$ under the shift action, which is a subset
of $\{\pm 1\}^{\mathbb{Z}}$ and hence a compact space. The orbit of
any of its members is dense in $\mathbb{X}$.  We thus have a
topological dynamical system $(\mathbb{X}, \mathbb{Z})$ that is
minimal. When equipped with the standard Borel $\sigma$-algebra, the
system admits a unique shift-invariant probability measure $\nu$, so
that the corresponding measure theoretic dynamical system
$(\mathbb{X},\mathbb{Z},\nu)$ is strictly ergodic \cite{BBG-K,BBG-Q}.

Any given $w\in\mathbb{X}$ is mapped to a signed Dirac comb $\omega$
via
\[
    \omega  \, = \sum_{n\in\mathbb{Z}} w (n) \,\delta_{n}\ts .
\] 
The image of $\mathbb{X}$ is a space of translation 
bounded measures that is compact
in the vague topology.  We inherit strict ergodicity via conjugacy,
and thus obtain an autocorrelation of the form of
Eq.~\eqref{BBG-eq:autodef} with coefficients $\eta(m)$ as in
Eq.~\eqref{BBG-eq:eta}. In fact, this autocorrelation does not depend
on the choice of the element from $\mathbb{X}$, so that we may choose
the fixed point $w$ from above for the concrete analysis. Due to the
nature of $w$,  the coefficients can alternatively be expressed
as
\[
     \eta(m)\, = \lim_{N\to\infty}\,
    \frac{1}{N}\, \sum_{n=0}^{N-1} v_n \, v_{n+m}
\]
for $m\ge 0$, together with $\eta(-m)=\eta(m)$. It is clear that
$\eta(0)=1$, and the scaling relations of $v$ lead to the
recursions \cite{BBG-K}
\begin{equation}\label{BBG-eq:tmrec}
\begin{split}
    \eta(2m) & \, =\,  \eta(m)\qquad \text{and} \\
    \eta(2m\!+\!1) & \, =\,  
    -\tfrac{1}{2} \bigl( \eta(m) + \eta(m\!+\!1)\bigr),
\end{split}
\end{equation}
which are valid for all $m\in\mathbb{Z}$. In particular, the second
relation, used with $m=0$, implies $\eta(1) = -\tfrac{1}{3}$, which
can also be calculated directly.

Since $\eta \! : \, \mathbb{Z} \longrightarrow \mathbb{C}$ is a
positive definite function with $\eta(0)=1$, there is a unique
probability measure $\mu$ on the unit circle (which we identify with
the unit interval here) such that
\begin{equation}\label{BBG-eq:HB}
     \eta (m) \, =  \int_{0}^{1} \mathrm{e}^{2 \pi \mathrm{i} m y} 
         \, \mathrm{d} \mu (y) \, ,
\end{equation}
which is a consequence of the Herglotz--Bochner theorem
\cite[Thm.~I.7.6]{BBG-Katz}. Since $\omega$ is supported on
$\mathbb{Z}$, the corresponding diffraction measure $\widehat{\gamma}$
is $1$-periodic, which follows from \cite[Thm.~1]{BBG-B02}; see also
\cite[Sec.~10.3.2]{BBG-TAO}. One then finds the relation
\[
    \widehat{\gamma} \, = \, \mu * \delta^{}_{\mathbb{Z}} 
\]
with the measure $\mu$ from Eq.~\eqref{BBG-eq:HB}, appropriately
interpreted as a measure on $[0,1)$ and hence also on
$\mathbb{R}$. Clearly, one also has $\mu =
\widehat{\gamma}|^{}_{[0,1)}$.  One can now analyse the spectral type
of $\widehat{\gamma}$ via that of the finite measure $\mu$, where we
follow \cite{BBG-K}; see also \cite{BBG-Q,BBG-BLvE}.

Defining $\varSigma(N) = \sum_{m=-N}^{N} \bigl(\eta(m)\bigr)^2$, a
two-step calculation with the recursion \eqref{BBG-eq:tmrec}
establishes the inequality $\varSigma (4N) \le \frac{3}{2} \varSigma
(2 N)$ for all $N\in\mathbb{N}$.  This implies $\lim_{N\to\infty}
\varSigma(N)/N=0$, wherefore Wiener's criterion
\cite{BBG-Wie27,BBG-Katz}, see also \cite[Prop.~8.9]{BBG-TAO}, tells
us that $\mu$ is a continuous measure, so that $\widehat{\gamma}$
cannot have any pure point component. Note that the absence of the
`trivial' pure point component of $\widehat{\gamma}$ on $\mathbb{Z}$
is due to the use of balanced weights, in the sense that $1$ and $-1$
are equally frequent. Consequently, the average weight is zero, and
the claim follows from \cite[Prop.~9.2]{BBG-TAO}.

Let us now define the distribution function $F$ by $F(x) = \mu \bigl(
[0,x] \bigr)$ for $x\in[0,1]$, which is a continuous function that
defines a Riemann--Stieltjes measure \cite[Ch.~X]{BBG-Lang}, so that
$\mathrm{d} F=\mu$. The recursion relation for $\eta$ now implies
\cite{BBG-K} the two functional relations
\[
        \mathrm{d} F \bigl( \tfrac{x}{2} \bigr) \pm
            \mathrm{d} F \bigl( \tfrac{x+1}{2} \bigr) \, = \,
        \left\{\begin{smallmatrix}1 \\
         -\cos(\pi x)\end{smallmatrix}\right\} \,
          \mathrm{d} F (x) \, ,
\]
which have to be satisfied by the $\textsf{ac}$ and $\textsf{sc}$
parts of $F$ separately, because we have $\mu^{}_{\textsf{ac}} \perp
\mu^{}_{\textsf{sc}}$ in the measure-theoretic sense; see
\cite[Thm.~I.20]{BBG-RS} or \cite[Thm.~VII.2.4]{BBG-Lang}. Therefore,
defining
\[ 
   \eta_\mathsf{ac} (m)\, = \int_{0}^{1} 
    \mathrm{e}^{2 \pi \mathrm{i} m x}\, \mathrm{d}F_{\textsf{ac}}(x)\, ,
\]
we know that the coefficients $\eta^{}_{\mathsf{ac}}(m)$ must satisfy
the same recursions \eqref{BBG-eq:tmrec} as $\eta(m)$, possibly with a
different initial condition $\eta^{}_{\mathsf{ac}}(0)$. The classic
Riemann--Lebesgue lemma \cite[Thm.I.2.8]{BBG-Katz} states that
$\lim_{m\to\pm\infty} \eta_\mathsf{ac} (m) = 0$. But this limit is
only compatible with $\eta_\mathsf{ac}(0)=0$, because
$\eta_\mathsf{ac}(1)=-\frac{1}{3}\eta_\mathsf{ac}(0)$ and
$\eta_\mathsf{ac}(2m)=\eta_\mathsf{ac}(m)$ for all $m\in\mathbb{N}$,
so that we must have $\eta_\mathsf{ac}\equiv 0$. This means
$F_{\mathsf{ac}}=0$ by the Fourier uniqueness theorem, wherefore $\mu$
and hence $\widehat{\gamma}$ (neither of which is the zero measure)
are purely singular continuous. The resulting distribution function
$F$ is illustrated in Figure~\ref{BBG-fig:tm}. Note that $F$ can
consistently be extended to a continuous function on $\mathbb{R}$ via
$F(x+n) = F(x) + n$ for $n\in\mathbb{Z}$ and then defines
$\widehat{\gamma}$ via $\dd F = \widehat{\gamma}$ in the
Lebesgue--Stieltjes sense. The function $F$ can efficiently be
calculated by means of the uniformly converging Volterra iteration
\begin{equation}\label{BBG-eq:Volt}
     F^{}_{n+1} (x) = \frac{1}{2} \int_{0}^{2x}
     \bigl( 1 - \cos (\pi y)\bigr) F^{\,\prime}_{n} (y)\,
     \mathrm{d} y  
\end{equation}
with $F^{}_{0} (x) = x$. In contrast to the Devil's staircase of
Figure~\ref{BBG-fig:cantor}, the TM distribution function is
\emph{strictly} increasing, which means that there is no plateau
(which would indicate a gap in the support of $\widehat{\gamma}$); see
\cite{BBG-BG08,BBG-TAO} and references therein for details and further
properties of $F$. So far, we have obtained the following result.

\begin{figure}
\begin{center}
\includegraphics[width=0.77\textwidth]{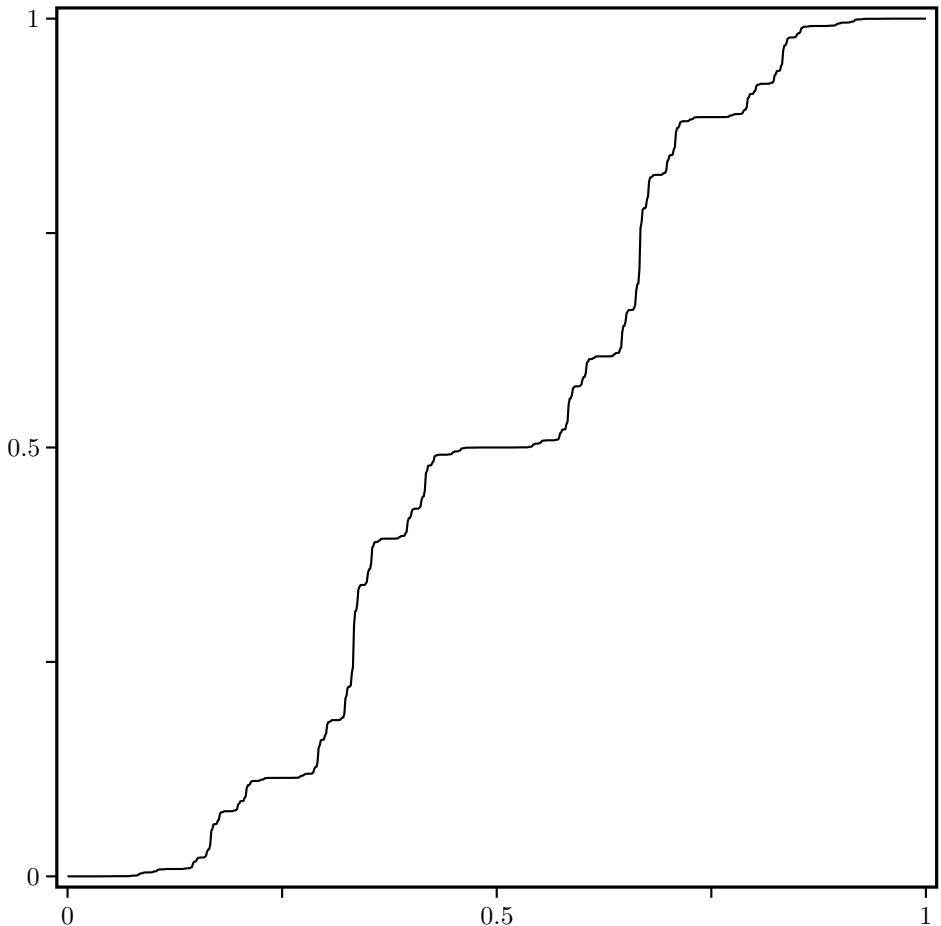}
\end{center}
\caption{\label{BBG-fig:tm}The strictly increasing distribution
  function of the classic, purely singular continuous TM measure on
  $[0,1]$.}
\end{figure}

\begin{thm}\label{BBG-thm:TM}
  Let\/ $w$ be any element of the Thue--Morse hull\/
  $\mathbb{X}=\mathbb{X}^{}_{\mathrm{TM}}$, the latter represented as
  a closed subshift of\/ $\{\pm 1 \}^{\mathbb{Z}}$, and consider the
  corresponding Dirac comb\/ $w \hspace*{1pt} \delta^{}_{\mathbb{Z}}$.
  Then, its autocorrelation\/ $\gamma$ exists and is given by\/
  $\gamma = \eta \hspace*{1pt} \delta^{}_{\mathbb{Z}}$ with\/ $\eta$
  being defined by Eq.~\eqref{BBG-eq:tmrec} together with the initial
  condition\/ $\eta(0) = 1$.
 
  The diffraction measure is\/ $\widehat{\gamma} = \mu *
  \delta^{}_{\mathbb{Z}}$, where\/ $\mu$ is the purely
  singular continuous probability measure from Eq.~\eqref{BBG-eq:HB}.
  In particular,  $\widehat{\gamma}$ is purely singular continuous 
  as well.      \qed
\end{thm}

To go one step further, Eq.~\eqref{BBG-eq:Volt} defines an iteration
sequence of distribution functions for absolutely continuous measures
that converges towards the TM measure in the vague topology. Writing
$\,\mathrm{d} F_{n} (x) = f_{n} (x) \, \mathrm{d} x$, one finds
\[
      f_{n} (x) \, = \, \prod_{m=0}^{n-1} 
          \bigl( 1 - \cos (2^{m+1} \pi x)\bigr) ,
\]
which, in the vague limit as $n\to\infty$, gives the well-known Riesz
product representation of the TM measure; compare \cite{BBG-Q} for details
and \cite{BBG-Z} for general background on Riesz products.

The TM sequence is closely related to the limit-periodic \emph{period
  doubling} (pd) sequence, compare \cite{BBG-BGG12,BBG-TAO} and
references therein, via the (continuous) sliding block map defined
by
\begin{equation} \label{BBG-eq:block-map}
  \phi \nts : \quad 1\bar{1} ,  \bar{1}1 \mapsto a
  \, , \quad 11  ,  \bar{1}\bar{1} \mapsto b \, ,
\end{equation}
which results in an exact 2-to-1 surjection from the hull
$\mathbb{X}^{}_{\mathrm{TM}}$ to $\mathbb{X}^{}_{\mathrm{pd}}$. The
latter is the hull of the period doubling substitution defined by
\begin{equation}\label{BBG-eq:pd-def}
   \varrho^{}_{\mathrm{pd}} \! : \quad
   a \mapsto ab \, , \quad b \mapsto aa\, .
\end{equation}
Viewed as topological dynamical systems, this means that
$(\mathbb{X}^{}_{\mathrm{pd}},\mathbb{Z})$ is a factor of
$(\mathbb{X}^{}_{\mathrm{TM}},\mathbb{Z})$. Since both are strictly
ergodic, this extends to the corresponding measure-theoretic dynamical
systems.  The period doubling sequence can be described as a regular
model set with a $2$-adic internal space \cite{BBG-BMS,BBG-Crelle} and
is thus pure point diffractive. This pairing also explains a
phenomenon observed in \cite{BBG-EM}, namely that the \emph{dynamical}
spectrum of the TM system is richer than its diffraction spectrum. By
the dynamical (or von Neumann) spectrum, we mean the spectrum of the
unitary operator induced by the shift on the Hilbert space $L^2
(\mathbb{X}, \nu)$, where $\nu$ is the unique shift-invariant
probability measure on $\mathbb{X}$; see \cite{BBG-Q} for more. Here,
the pure point part of the dynamical spectrum is the ring $\mathbb{Z}
[\frac{1}{2}]$, which is not even finitely generated (and only the
`trivial' part $\mathbb{Z}$ is detected by the diffraction measure of
the TM system with general weights). In fact, our above measure $\mu$
from Theorem~\ref{BBG-thm:TM} represents the maximal spectral measure
in the ortho-complement of the pure point sector
\cite{BBG-Q,BBG-BLvE}. The missing pure point part, however, is fully
recovered via the \emph{diffraction} measure of
$\mathbb{X}_{\mathrm{pd}}$; see \cite{BBG-TAO} for details and
\cite{BBG-BLvE} for a general discussion of this phenomenon.

Various generalisations of this result are known by now. First of all,
and perhaps not surprisiningly, this generalises to an entire family
of bijective, binary substitutions \cite{BBG-BGG12}. Moreover,
extensions to higher dimensions are also possible, including the
explicit nature of the resulting diffraction measure; compare
\cite{BBG-F05,BBG-squiral} and references therein.

\subsection{Rudin--Shapiro Sequence}

The (binary) Rudin--Shapiro (RS) chain is a bi-infinite deterministic
sequence, with polynomial (in fact linear) complexity function and
thus zero entropy. It can be described recursively as $w =
(w(n))^{}_{n\in\mathbb{Z}}$ with $w(n)\in\{\pm 1\}$, with initial
conditions $w(-1)=-1$, $w(0)=1$ and the recursion
\begin{equation}\label{BBG-eq:rs}
   w(4n+\ell) \, = \,
    \begin{cases} w(n),  & \mbox{for $\,\ell\in\{0,1\}$,} \\
          (-1)^{n+\ell}\,w(n), & \mbox{for $\,\ell\in\{2,3\}$,}
     \end{cases}
\end{equation}
which determines $w(n)$ for all $n\in\mathbb{Z}$. The orbit closure of
$w$ under the shift action is the (discrete) RS hull
$\mathbb{X}^{}_{\mathrm{RS}}$. Alternatively, one can start from a
primitive substitution on a $4$-letter alphabet (via $a\mapsto ac$,
$b\mapsto dc$, $c\mapsto ab$ and $d\mapsto db$) and define a
quaternary hull, which then maps to the binary hull via a simple
reduction to two letters (for instance via $a,c\mapsto 1$ and
$b,d\mapsto -1$); compare \cite{BBG-AS,BBG-Q} or
\cite[Sec.~4.7.1]{BBG-TAO} for details. The two hulls define
topologically conjugate dynamical systems, with local derivation rules
in both directions; see \cite[Rem.~4.11]{BBG-TAO}.

The shift action on $\mathbb{X}^{}_{\mathrm{RS}}$ is strictly ergodic,
so that one can define functions $\eta, \vartheta \! : \,
\mathbb{Z} \longrightarrow \mathbb{C}$ via
\[
\begin{split}
    \eta (m) &\, = \, \lim_{N\to\infty} \frac{1}{2N+1}
        \sum_{n=-N}^{N} w(n) \, w(n-m) \quad \text{and} \\
    \vartheta (m) &\, = \, \lim_{N\to\infty} \frac{1}{2N+1}
        \sum_{n=-N}^{N} (-1)^{n} w(n) \, w(n-m) \ts ,
\end{split}
\]
where all limits exist due to unique ergodicity (which is best
formulated on the level of the $4$-letter alphabet mentioned
above). In particular, one finds $\eta (0) = 1$ and $\vartheta (0) =
0$. The recursive structure of Eq.~\eqref{BBG-eq:rs} now implies the
validity of a closed set of recursive equations
\cite{BBG-BG09,BBG-BG11}, namely
\[
   \begin{split}
   \eta(4m)& \,=\, \tfrac{1+(-1)^m}{2}\,\eta(m) ,\\[0.5ex]
   \eta(4m\!+\!1) &\,=\, \tfrac{1-(-1)^m}{4}\,\eta(m) +
   \tfrac{(-1)^m}{4}\,\vartheta(m) - \tfrac{1}{4}\,
   \vartheta(m\!+\!1) , \\[1.3ex]
   \eta(4m\!+\!2) &\,=\, 0,\\[0.5ex]
   \eta(4m\!+\!3) &\,=\, \tfrac{1+(-1)^m}{4}\,\eta(m\!+\!1) -
   \tfrac{(-1)^m}{4}\,\vartheta(m)
                 + \tfrac{1}{4}\vartheta(m\!+\!1),
   \end{split}
\]
together with
\[
   \begin{split}
   \vartheta(4m) &\,=\, 0,\\[0.5ex]
   \vartheta(4m\!+\!1) &\,=\, \tfrac{1-(-1)^m}{4}\,\eta(m) -
   \tfrac{(-1)^m}{4}\,\vartheta(m) +\tfrac{1}{4}\,
          \vartheta(m\!+\!1) , \\[0.5ex]
   \vartheta(4m\!+\!2)&\,=\, \tfrac{(-1)^m}{2}\,\vartheta(m)+
          \tfrac{1}{2}\,
   \vartheta(m\!+\!1) , \\[0.5ex]
   \vartheta(4m\!+\!3) &\,=\, -\tfrac{1+(-1)^m}{4}\,\eta(m\!+\!1) -
   \tfrac{(-1)^m}{4}\,\vartheta(m)
        + \tfrac{1}{4}\,\vartheta(m\!+\!1) .
   \end{split}
\]
which hold for all $m\in\mathbb{Z}$; see \cite[Sec.~10.2]{BBG-TAO} for
details. A careful inspection shows that the unique solution of this
set of equations, with the initial conditions mentioned above, is
$\eta(m) = \delta^{}_{m,0}$ together with $\vartheta(m) = 0$ for all
$m\in\mathbb{Z}$. Hence, despite the deterministic nature of the RS
sequence, the autocorrelation measure is simply given by
$\gamma^{}_{\mathrm{RS}} = \delta^{}_{0}$, so that $\widehat{
  \gamma^{}_{\mathrm{RS}}} = \lambda$, where $\lambda$ again denotes
Lebesgue measure. Alternatively, the result also follows from the
exposition in \cite{BBG-Q,BBG-PF}.

\begin{thm}
  Let\/ $w$ be any element of the Rudin--Shapiro hull\/
  $\mathbb{X}^{}_{\mathrm{RS}} \subset \{\pm 1 \}^{\mathbb{Z}}$, and
  consider the corresponding Dirac comb\/ $w \ts
  \delta^{}_{\mathbb{Z}}$. Then, its autocorrelation exists and is
  given by\/ $\gamma^{}_{\mathrm{RS}} = \delta^{}_{0}$, with
  diffraction measure\/ $\widehat{\gamma^{}_{\mathrm{RS}}} = \lambda$.
  \qed
\end{thm}

As in the case of the TM sequence, the non-trivial pure point part of
the dynamical spectrum (which is $\mathbb{Z} [\frac{1}{2}]$ once
again) is not `seen' by the diffraction measure, while $\lambda$ (with
multiplicity $2$) represents once again the maximal spectral measure
in the ortho-complement of the pure point sector.  However, the
missing pure point component can be recovered by a suitable factor
system, the latter obtained via the block map defined by
Eq.~\eqref{BBG-eq:block-map}. The corresponding factor is represented
by a limit-periodic substitution rule that is somewhat reminiscent of
the paper folding sequence \cite{BBG-AS}; see
\cite[Sec.~10.2]{BBG-TAO} for a complete discussion and
\cite{BBG-BLvE} for the general connection between dynamical and
diffraction spectra. The structure underlying the RS sequence can be
generalised to higher-dimensional lattice substitutions in a rather
systematic way; see \cite{BBG-F03} for details.

\subsection{Bernoullisation}

Let us begin this discussion by recalling the structure of the full
Bernoulli shift from the viewpoint of kinematic diffraction. The
classic coin tossing process leads to the Dirac comb
\[
   \omega \, = \,
   \sum_{n\in\mathbb{Z}} X(n) \, \delta_{n} \, ,
\]
where the $(X(n))^{}_{n\in\mathbb{Z}}$ form an i.i.d.\ family of
random variables, each taking values $1$ and $-1$ with probabilities
$p$ and $1-p$, respectively.  By an application of the strong law of
large numbers (SLLN, see \cite{BBG-Ete81} for a favourable formulation),
almost every realisation has the autocorrelation measure
\[
    \gamma \, = \, (2p-1)^{2} \, \delta^{}_{\mathbb{Z}}
    + 4 p (1-p) \, \delta^{}_{0} \, ,
\]
and hence (via Fourier transform) the diffraction measure
\[
    \widehat{\gamma} \, = \, (2p-1)^{2} \, \delta^{}_{\mathbb{Z}}
    + 4 p (1-p) \, \lambda \, .
\]
Here, we have used the classic Poisson summation formula
$\widehat{\delta^{}_{\mathbb{Z}}} = \delta^{}_{\mathbb{Z}}$; compare
\cite{BBG-BG11} and references therein, as well as
\cite[Sec.~9.2]{BBG-TAO} for a formulation in the diffraction
context.  When $p=\frac{1}{2}$, the diffraction boils down to
$\widehat{\gamma} = \lambda$. Here, the point part is extinct because
the average scattering strength vanishes. For proofs, we refer the
reader to \cite{BBG-BM98,BBG-BBM10}, while \cite{BBG-Kuel1,BBG-Kuel2}
contain several important and non-trivial generalisations and
extensions; see also \cite{BBG-DL} for important related material.

The Bernoulli chain has (metric) entropy \cite{BBG-CFS,BBG-EW}
\[
   h(p) \, = \, - p\log (p) - (1\!-\!p) \log (1\!-\!p) ,
\]
which is maximal for $p=\frac{1}{2}$, with $h(\frac{1}{2})=\log
(2)$. It vanishes for the deterministic limiting cases
$p\in\{0,1\}$. For the latter, we have
$\omega=\mp\delta^{}_{\mathbb{Z}}$, and consequently obtain the pure
point diffraction measure $\widehat{\gamma} = \delta^{}_{\mathbb{Z}}$,
again via Poisson's summation formula.

Now, the theory of random variables allows for an interpolation
between deterministic (binary) sequences and coin tossing sequences as
follows. If $w \in \{ \pm 1 \}^{\mathbb{Z}}$ denotes a deterministic
sequence (which we assume to be uniquely ergodic for simplicity),
consider the random Dirac comb \cite{BBG-BG09}
\begin{equation}\label{BBG-eq:randrs}
   \omega_{p} \, = \sum_{n\in\mathbb{Z}} w(n)\ts X(n)\, \delta_{n}\, ,
\end{equation}
where $(X(n))^{}_{n\in\mathbb{Z}}$ is, as above, an i.i.d.\ family of
random variables with values in $\{\pm 1\}$ and probabilities $p$ and
$1-p$. This `Bernoullisation' of $w$ can be viewed as a `model of
second thoughts', where the sign of the weight at position $n$ is
changed with probability $1-p$; compare \cite[Sec.~11.2.2]{BBG-TAO}.

Let $w$ now be the Rudin--Shapiro sequence from above.  By a (slightly
more complicated) application of the SLLN, it can be shown
\cite{BBG-BG09} that the autocorrelation $\gamma^{}_{p}$ of the Dirac
comb $\omega^{}_{p}$ is then almost surely given by
\[
   \gamma^{}_{p}  \, = \, (2p-1)^{2}\,\gamma^{}_{\mathrm{RS}} + 
   4 p (1-p)\, \delta^{}_{0}
   \, = \, \delta^{}_{0}\, ,
\]
irrespective of the value of the parameter $p\in [0,1]$. Recall that
two measures with the same autocorrelation are called
\emph{homometric}; see \cite[Sec.~9.6]{BBG-TAO} for background.  Our
observation thus establishes the following classic result; see
\cite{BBG-BG09,BBG-BG11,BBG-TAO} for details.

\begin{thm}
  The random Dirac combs\/ $\omega_{p}$ of Equation~\eqref{BBG-eq:randrs}
  with real parameter values\/ $p\in [0,1]$ are\/ $(\nts\nts$almost surely$)$
  homometric, with absolutely continuous diffraction measure\/
  $\!\widehat{\,\gamma_{p}\,}\! = \widehat{\gamma^{}_{\mathrm{RS}}} =
  \lambda$, irrespective of the value of\/ $p$. In other words, the
  family\/ $\big\{\omega^{}_{p} \mid {p\in [0,1]}\big\}$ is\/
  $(\nts\nts$almost surely$)$ isospectral.  \hfill \qed
\end{thm}

This result shows that diffraction can be insensitive to entropy,
because the family of Dirac combs $\omega_{p}$ of
Eq.~\eqref{BBG-eq:randrs} continuously interpolates between the
deterministic Rudin--Shapiro case with zero entropy and the completely
random Bernoulli chain with maximal entropy $\log(2)$. Clearly, the
Bernoullisation procedure can be applied to other sequences as well,
and can be generalised to higher dimensions. For further aspects of
entropy versus diffraction, we refer to
\cite{BBG-BG09,BBG-BG12,BBG-BLR}.

\subsection{Random Dimers on the Line}

Another instructive example \cite{BBG-BE11} is based on certain dimer
configurations on $\mathbb{Z}$. To formulate it, we follow the
exposition in \cite{BBG-BG12} and partition $\mathbb{Z}$ into a
close-packed arrangement of `dimers' (pairs of neighbours), without
gaps or overlaps. Clearly, there are just two possibilities to do so,
because the position of the first dimer fixes that of all
others. Next, decorate each dimer randomly with either $(1,-1)$ or
$(-1,1)$, with equal probability. This results in patches such as
\[
\begin{split}
 \dots [+ \,\, -]\nts\nts [-\,\,  +]\nts\nts [-\,\,  +]\nts\nts 
       [+\,\,  -]\nts &\nts [-\,\, +]\nts\nts [-\,\, +]\nts\nts 
      [-\,\, +]\nts\nts [+\,\, -]\nts\nts [+\,\, -]\dots\\
      \dots [-\,\, +]\nts\nts [+\,\, -]\nts\nts [+\,\, -]\nts\nts 
      [-\,\, +] \nts\nts
      [+\, &\, -]\nts\nts [+\,\, -]\nts\nts [+\,\, -]\nts\nts 
      [-\,\, +]\nts\nts [-\,\, +]\nts\nts [+\,\, -]\dots
\end{split}
\]
where the dimer boxes are indicated by brackets. The set of all
decorated sequences defined in this way is given by
\[
   \mathbb{X} \, = \, \bigl\{ w \in \{ \pm 1 \} ^{\mathbb{Z}} \mid
   M(w) \subset 2 \mathbb{Z} \,\text{ or }
   M(w) \subset 2 \mathbb{Z} + 1 \bigr\} \, ,
\]
where $M(w) := \{ n\in\mathbb{Z} \mid w(n) = w(n+1) \}$. Note that
$M(w)$ is empty precisely for the two periodic sequences that are
defined by $w(n)=\pm (-1)^n$ for $n\in\mathbb{Z}$. Clearly, 
$\mathbb{X}\subset\{\pm 1\}^{\mathbb{Z}}$ is closed and hence compact.

Let $w\in \mathbb{X}$ and consider the corresponding signed Dirac comb
on $\mathbb{Z}$ with binary weights $w(n)\in\{\pm 1 \}$. One can then show
(again via the SLLN) that the corresponding autocorrelation almost
surely exists and is given by \cite{BBG-BE11}
\begin{equation}\label{BBG-eq:dms-auto}
   \gamma\, =\, \delta^{}_{0} - \frac{1}{2} 
              (\delta^{}_{1} + \delta^{}_{-1})\, .
\end{equation}
The corresponding diffraction measure is then
\begin{equation}\label{BBG-eq:dms}
   \widehat{\gamma} \,=\,
   \bigl( 1 - \cos(2 \pi k) \bigr) \lambda\, ,
\end{equation}
which is again purely absolutely continuous. Here, the (smooth)
Radon--Nikodym density relative to $\lambda$ is written as a function
of $k$. Note that the diffraction measure for general weights
$h_{+}$ and $h_{-}$ is given by
\[
   \widehat{\gamma^{}_{\pm}} \, = \, 
   \frac{\lvert h_{+} + h_{-} \rvert^{2}}{4} 
   \, \delta^{}_{\mathbb{Z}} +  \frac{\lvert h_{+} - h_{-} \rvert^{2}}{4}
   \, \widehat{\gamma}
\]
with $\widehat{\gamma}$ as in Eq.~\eqref{BBG-eq:dms}. In particular,
the measure $\widehat{\gamma^{}_{\pm}}$ shows only the `trivial' pure
point diffraction contribution that arises as the consequence of
$\mathbb{Z}$ being the support of the weighted measure under
consideration.  The same phenomenon also occurs for general
(non-balanced) TM and RS sequences; compare \cite[Rems.~10.3 and
10.5]{BBG-TAO}.

On first sight, the system looks disordered, with entropy $\frac{1}{2}
\log (2)$. This seems (qualitatively) reflected by the
diffraction. However, the system also defines a measure-theoretic
dynamical system under the action of $\mathbb{Z}$, as generated by the
shift. As such, it has a dynamical spectrum that does contain a pure
point part, with eigenvalues $0$ and $\frac{1}{2}$; we refer to
\cite{BBG-Q} for general background on this concept, and to
\cite{BBG-BE11} for the actual calculation of the eigenfunctions. The
extension to a (continuous) dynamical system $\mathbb{X}_{\mathrm{c}}$
under the general translation action of $\mathbb{R}$ is done via
suspension; see \cite[Ch.~11.1]{BBG-CFS} (where the suspension is
called a special flow) or \cite{BBG-EW} for general background.

This finding suggests that some degree of order must be present that
is neither visible from the entropy calculation nor from the
diffraction measure alone. Indeed, in analogy with the situation of
the TM and the RS sequence, one can define a factor of the system by a
sliding block map $\phi\!:\, \mathbb{X}\longrightarrow \{\pm
1\}^{\mathbb{Z}}$ defined by $(\phi w)(n) = -w(n)w(n+1)$. It maps
$\mathbb{X}$ globally 2:1 onto
\[
  \mathbb{Y}=\phi(\mathbb{X})=\bigl\{v\in\{\pm 1\}^{\mathbb{Z}}\mid
  \mbox{$v(n)=1$ for all $n\in2\mathbb{Z}$ or for all 
         $n\in 2\mathbb{Z}+1$}\bigr\}.
\]
The suspension $\mathbb{Y}_{\!\mathrm{c}}$ (for the action of $\mathbb{R}$) 
is defined as above. The mapping $\phi$ extends accordingly.

The autocorrelation and diffraction measures of the signed Dirac comb
$v\delta^{}_{\mathbb{Z}}$ for an element $v\in\mathbb{Y}$ are almost
surely given by
\[
    \gamma \, = \, \frac{1}{2}\delta^{}_{0} + 
    \frac{1}{2}\delta^{}_{2\mathbb{Z}}
    \quad\text{and}\quad
    \widehat{\gamma}\, =\, \frac{1}{2}\lambda + 
    \frac{1}{4}\delta^{}_{\mathbb{Z}/2}\, .
\]
The diffraction of the factor system $\mathbb{Y}$ uncovers the
`hidden' pure point part of the dynamical spectrum, which was absent
in the purely absolutely continuous diffraction of the signed Dirac
comb $w\hspace{1pt} \delta^{}_{\mathbb{Z}}$ with $w\in\mathbb{X}$. In
summary, we have the following situation \cite{BBG-BE11,BBG-BLvE}.

\begin{thm}
  The diffraction measure of the close-packed dimer system\/
  $\mathbb{X}$ with balanced weights is purely absolutely continuous
  and given by Eq.~\eqref{BBG-eq:dms}, which holds almost surely relative
  to the natural invariant measure of the system.

  The dynamical spectrum of the continuous close-packed dimer system\/
  $\mathbb{X}_{\mathrm{c}}$ under the translation action of\/
  $\mathbb{R}$ contains the pure point part\/ $\mathbb{Z}/2$ together
  with a countable Lebesgue spectrum.

  The non-trivial part\/ $\mathbb{Z} + \frac{1}{2}$ of the dynamical
  point spectrum is not reflected by the diffraction spectrum of\/
  $\mathbb{X}_{\mathrm{c}}$, but can be recovered via the diffraction
  spectrum of a suitable factor, such as\/ $\mathbb{Y}_{\!\mathrm{c}}$.
  \qed
\end{thm}

As in the case of the Thue--Morse system, where the missing pure point
part of the dynamical spectrum is recovered by the diffraction measure
of the period doubling factor, we thus see that and how we can recover
the missing eigenvalue via a generalised $2$-point function.  This
observation can be extended to symbolic systems over finite
alphabets and also to uniquely ergodic Delone dynamical systems of
finite local complexity; see \cite{BBG-BLvE} for details.

\subsection{Ledrappier's Shift Space}

For a long time, people had expected that higher dimensions
are perhaps more difficult, but not substantially different.
This turned out to be a false premise though, as can be
seen from the now classic monograph \cite{BBG-Sch}.

In our present context, we pick one characteristic example, the system
due to Ledrappier \cite{BBG-L}, to demonstrate a new phenomenon. We
follow the brief exposition in \cite{BBG-BG12} and consider a specific
subset of the full shift space $\{\pm 1\}^{\mathbb{Z}^{2}}$, defined
by
\begin{equation}\label{BBG-eq:def-L}
  \mathbb{X}_{\mathrm{L}} \, =\,
  \bigl\{ w \in \{\pm 1\}^{\mathbb{Z}^{2}} \! \mid
     w(x)\, w(x+e^{}_{1})\, w(x+e^{}_{2}) = 1 \,
     \mbox{ for all }\, x \in \mathbb{Z}^{2} \bigr\},
\end{equation}
where $e^{}_{1}$ and $e^{}_{2}$ denote the standard Euclidean basis
vectors in the plane. On top of being a closed subshift,
$\mathbb{X}_{\mathrm{L}}$ is also an Abelian group (here written
multiplicatively), which then comes with a unique, normalised Haar
measure. The latter is also shift invariant, and the most natural
measure to be considered in our context; see also the reformulation
in terms of Gibbs (or equilibrium) measures in \cite{BBG-Slawny}.

The system is interesting because the number of patches of a given
radius (up to translations) grows exponentially in the radius rather
than in the area of the patch. This phenomenon is called \emph{entropy
  of rank $1$}, and indicates a new class of systems in higher
dimensions. More precisely, along any lattice direction of
$\mathbb{Z}^{2}$, the linear subsystems essentially behave like
one-dimensional Bernoulli chains. It is thus not too surprising that
the diffraction measure satisfies the following theorem, though its
proof \cite{BBG-BW10} has to take care of the special directions connected
with the defining relations of $\mathbb{X}_{\mathrm{L}}$.

\begin{thm}
  If\/ $w$ is an element of the Ledrappier subshift\/
  $\mathbb{X}_{\mathrm{L}}\!$ of Eq.~\eqref{BBG-eq:def-L}, the
  corresponding weighted Dirac comb\/ $w\ts \delta^{}_{\mathbb{Z}^2}$ has
  diffraction measure\/ $\lambda$, which holds almost surely relative to
  the Haar measure of\/ $\mathbb{X}_{\mathrm{L}}$.  \qed
\end{thm}

So, the Ledrappier system is homometric to the (full) Bernoulli shift
on $\{ \pm 1 \}^{\mathbb{Z}^2}$, which means that an element of either
system almost surely has diffraction measure $\lambda$. As mentioned
before, via a suitable product of two Rudin--Shapiro chains, also a
deterministic system with diffraction $\lambda$ exists. This clearly
demonstrates the insensitivity of pair correlations to the (entropic)
type of order or disorder in the underlying system; see also
\cite{BBG-BG09}. Due to the defining relation in
Eq.~\eqref{BBG-eq:def-L}, it is clear that certain three-point
correlations in the Ledrappier system cannot vanish, and thus make it
distinguishable from the Bernoulli shift.

Although correlation functions of third order can resolve the
situation in this case (and in many other examples as well
\cite{BBG-DM09,BBG-LM09}), one can consider other dynamical systems
(such as the ($\times 2,\times 3$)-shift \cite{BBG-BW10}) that share
almost all correlation functions with the Bernoulli shift on
$[0,1]^{\mathbb{Z}^{2}}$. This is a clear indication that our present
understanding of `order' is incomplete, and that we still lack a good
set of tools for the detection and classification of order. For a
recent alternative based on direct space statistics, we refer to
\cite{BBG-BGHJ}.

\subsection{Random Matrix Ensembles}

Another interesting class of random point sets derives from the
(scaled) eigenvalue distribution of certain random matrix ensembles;
see \cite{BBG-BK} and references therein. The global eigenvalue
distribution of random orthogonal, unitary or symplectic matrix
ensembles is known to asymptotically follow the classic semi-circle
law. More precisely, this law describes the eigenvalue distribution of
the underlying ensembles of symmetric, Hermitian or (symplectically)
self-dual matrices with Gaussian distributed entries.  The
corresponding random matrix ensembles are called GOE, GUE and GSE,
with attached $\beta$-parameters $1$, $2$ and $4$, respectively. They
permit an interpretation as a Coulomb gas, where $\beta$ is the power
in the central potential; see \cite{BBG-AGZ,BBG-Mehta} for general
background and \cite{BBG-D,BBG-F} for the results that are relevant
here.

{}For matrices of dimension $N$, the semi-circle has radius
$\sqrt{2N/\pi}$ and area $N$. Note that, in comparison with
\cite{BBG-Mehta}, we have rescaled the density by a factor
$1/\sqrt{\pi}$, so that we really have a semi-circle (and not a
semi-ellipse). To study the local eigenvalue distribution for
diffraction, we rescale the central region (between $\pm 1$, say) by
$\sqrt{2N/\pi} $. This leads, in the limit as $N\to\infty$, to an
ensemble of point sets on the line that can be interpreted as a
stationary, ergodic point process of intensity $1$; for $\beta=2$, see
\cite[Ch.~4.2]{BBG-AGZ} and references therein for details.  Since the
underlying process is simple (meaning that, almost surely, no point is
occupied twice), almost all realisations are point sets of density
$1$.

It is possible to calculate the autocorrelation of these processes, on
the basis of Dyson's correlation functions \cite{BBG-D}.  Though these
functions originally apply to the circular ensembles, they have been
adapted to the other ensembles by Mehta \cite{BBG-Mehta}.  For all
three ensembles mentioned above, this leads to an autocorrelation of
the form
\begin{equation}\label{BBG-eq:Dyson-auto}
    \gamma \, = \, \delta_{0} + 
    \bigl( 1 - f(\lvert x \rvert ) \bigr) \lambda 
\end{equation}
where $f$ is a locally integrable function that depends on $\beta$;
see \cite{BBG-BK} for the explicit formulas, and the left panel of
Figure~\ref{BBG-fig:dyson} for an illustration.

The diffraction measure is the Fourier transform of $\gamma$, which
has also been calculated in \cite{BBG-D,BBG-Mehta}. Recalling
$\widehat{\delta_{0}} =\lambda$ and $\widehat{\lambda} = \delta_{0}$,
the result is always of the form
\begin{equation}\label{BBG-eq:Dyson-Fourier}
   \widehat{\gamma}\, = \,\delta_{0} + 
   \bigl( 1 - b(k) \bigr) \lambda \,=\,
   \delta_{0} + h(k) \, \lambda ,
\end{equation}
where $b = \widehat{f}$. The Radon--Nikodym density $h$ depends on
$\beta$ and is summarised in \cite{BBG-BK}.
Figure~\ref{BBG-fig:dyson} illustrates the result for the three
ensembles.

\begin{figure}
\centerline{\includegraphics[width=\textwidth]{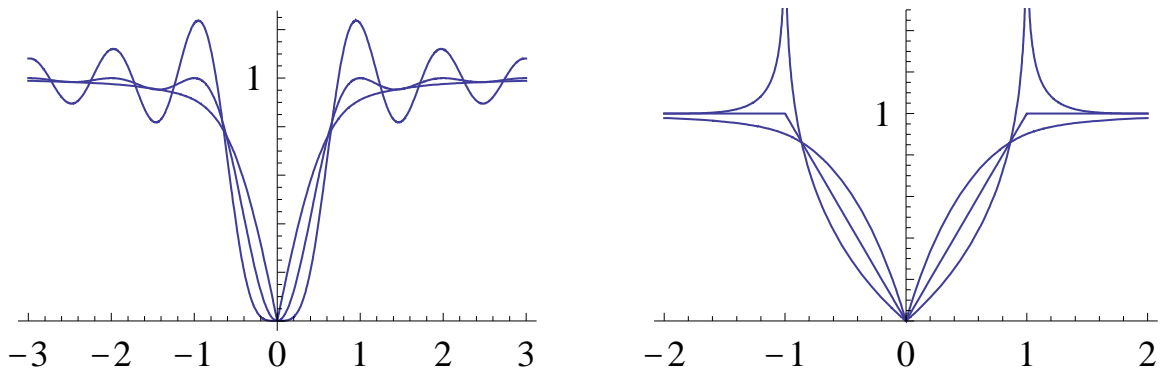}}
\caption{Absolutely continuous part of the autocorrelation (left) and
  the diffraction (right) for the three random matrix derived
  point set ensembles on the
  line, with $\beta\in\{1,2,4\}$. On the left, the oscillatory
  behaviour increases with $\beta$. On the right, $\beta=2$
  corresponds to the piecewise linear function with bends at $0$ and
  $\pm 1$, while $\beta=4$ shows a locally integrable singularity at
  $\pm 1$. The latter reflects the slowly decaying oscillations on
  the left.}  
\label{BBG-fig:dyson}
\end{figure}

A similar approach is possible on the basis of the eigenvalues of
general complex random matrices. This leads to the ensemble studied by
Ginibre \cite{BBG-Mehta}, which is also discussed in
\cite{BBG-BK}. One common feature of the resulting point sets is the
effectively repulsive behaviour of the points, which leads to the
`dip' around $0$ for $\widehat{\gamma}$.  For the two systems
mentioned in this section, we omit the formulation of the full results
and refer the reader to \cite{BBG-BK} for details. Further
developments around determinantal and related point processes are
described in reference \cite{BBG-BKM}.

\section{The Renewal Process}

A large and interesting class of processes in one dimension can be
described as a \emph{renewal} process \cite{BBG-F2,BBG-BBM10,BBG-BK}.
Here, one starts from a probability measure $\mu$ on $\mathbb{R}_{+}$
(the positive real line) and considers a machine that moves at
constant speed along the real line and drops a point on the line with
a waiting time that is distributed according to $\mu$.  Whenever this
happens, the internal clock is reset and the process resumes. Let us
(for simplicity) assume that both the velocity of the machine and the
expectation value of $\mu$ are $1$, so that we end up with
realisations that are, almost surely, point sets in $\mathbb{R}$ of
density $1$ (after we let the starting point of the machine move to
$-\infty$, say).
 
Clearly, the resulting process is stationary and can thus be analysed
by considering all realisations which contain the origin.  Moreover,
there is a clear (distributional) symmetry around the origin, so that
we can determine the corresponding autocorrelation $\gamma$ of almost
all realisations from studying what happens to the right of
$0$. Indeed, if we want to know the frequency per unit length of the
occurrence of two points at distance $x$ (or the corresponding
density), we need to sum the contributions that $x$ is the first point
after $0$, the second point, the third, and so on. In other words, we
almost surely obtain the autocorrelation
\begin{equation} \label{BBG-eq:auto-1}
    \gamma \, = \, \delta_{0} + \nu + \widetilde{\nu}
\end{equation}
with $\nu = \mu + \mu * \mu + \mu * \mu * \mu + \ldots$, where the
proper convergence of the sum of iterated convolutions follows from
\cite[Lemma~4]{BBG-BBM10} or from \cite[Sec.~11.3]{BBG-TAO}. Note that
the point measure at $0$ simply reflects the fact that the almost sure
density of the resulting point set is $1$. Indeed, $\nu$ is a
translation bounded positive measure, and satisfies the renewal
relations (compare \cite[Ch.~XI.9]{BBG-F2} or \cite[Prop.~1]{BBG-BBM10}
for a proof)
\begin{equation}\label{BBG-eq:ren-rel}
   \nu \, = \, \mu + \mu * \nu \qquad\text{and}\qquad
   (1-\widehat{\mu}\, )\, \widehat{\nu}
    \, = \, \widehat{\mu}\, ,
\end{equation}
where $\widehat{\mu}$ is a uniformly continuous and bounded function
on $\mathbb{R}$. The second equation emerges from the first by Fourier
transform, but has been rearranged to highlight the relevance of the
set $S=\{k \mid \widehat{\mu} (k) = 1 \}$ of singularities.  In this
setting, the measure $\gamma$ of Eq.~\eqref{BBG-eq:auto-1} is both
positive and positive definite.

Based on the structure of the support of the underlying probability
measure $\mu$, one can determine the diffraction of the renewal
process explicitly. To do so for a probability measure $\mu$ on
$\mathbb{R}_{+}$ with mean $1$, we assume the existence of
a moment of $\mu$ of order $1+\varepsilon$ for some 
$\,\varepsilon > 0$; we refer to \cite{BBG-BBM10} for details on this 
condition. The diffraction measure of the point set realisations
of the stationary renewal process based on $\mu$ almost
surely is of the form
\[
    \widehat{\gamma} \, = \, 
    \widehat{\gamma}^{}_{\mathsf{pp}}  + (1-h)\,\lambda ,
\]
where $h$ is a locally integrable function on $\mathbb{R}$ that is
continuous almost everywhere. The pure point part is trivial,
meaning $\widehat{\gamma} = \delta^{}_{0}$, unless the support
of $\mu$ is contained in a lattice. The details are stated below in
Theorem~\ref{BBG-thm:renewal}. Proofs of these claims as well as 
further results can be found in \cite{BBG-BBM10,BBG-BK,BBG-TAO}.

The renewal process is a versatile method to produce interesting point
sets on the line. These include random tilings with finitely many
intervals (which are Delone sets) as well as the homogeneous Poisson
process on the line (where $\mu$ is the exponential distribution with
mean $1$); see \cite[Sec.~3]{BBG-BBM10} for explicit examples and
applications.  In particular, if one employs a suitably normalised
version of the Gamma distribution, one can formulate a one-parameter
family of renewal processes that continuously interpolates between the
Poisson process (total positional randomness) and the lattice
$\mathbb{Z}$ (perfect periodic order); compare \cite[Ex.~3]{BBG-BBM10} for
more. The general result reads as follows.

\begin{thm}\label{BBG-thm:renewal} 
  Let\/ $\varrho$ be a probability measure on\/ $\mathbb{R}_{+}$ with
  mean\/ $1$, and assume that a moment of\/ $\varrho$ of order\/
  $1+\varepsilon$ exists for some\/ $\varepsilon > 0$. Then, the point
  sets obtained from the stationary renewal process based on\/
  $\varrho$ almost surely have a diffraction measure of the form
\[
    \widehat{\gamma} \, = \, 
    \widehat{\gamma}^{}_{\mathsf{pp}}  + (1-h)\,\lambda\, ,
\]
  where\/ $h$ is a locally integrable function on\/ $\mathbb{R}$ that is
  continuous except for at most countably many points\/ $($namely those
  of the set\/ $S=\{k\mid \widehat{\varrho}(k)=1\})$. On\/
  $\mathbb{R}\setminus S$, the function\/ $h$ is given by
\[
    h(k) \, = \, \frac{2\,\bigl(\lvert\widehat{\varrho}
    (k)\rvert^2 - \mathrm{Re} (\widehat{\varrho}(k))\bigr)}
    {\lvert 1 - \widehat{\varrho} (k)\rvert^2}\, .
\]
  Moreover, the pure point part is
\[
   \widehat{\gamma}^{}_{\mathsf{pp}} \, = \,
   \begin{cases} \delta^{}_{0} , & \text{if\/ $\mathrm{supp}
     (\varrho)$ is not a subset of a lattice}, \\
       \delta^{}_{\mathbb{Z}/b} , & \text{otherwise},
     \end{cases}
\]
  where\/ $b\mathbb{Z}$ is the coarsest lattice that contains\/
   $\mathrm{supp} (\varrho)$.    \qed
\end{thm}

In one dimension, the renewal process allows an efficient 
derivation of the diffraction of random tilings, which we
briefly summarise now.

\section{Random Tilings}

The deterministic Fibonacci chain can be defined by the primitive
substitution rule $a \mapsto ab$, $b \mapsto a$, which defines
a strictly ergodic (discrete) hull. When $a$ and $b$ are replaced
by intervals of length $\tau = \frac{1}{2} (1 + \sqrt{5}\,)$ and
$1$, respectively, the left endpoints of the intervals define a
model set (or cut and project set). The corresponding Dirac comb
leads to the pure point diffraction measure
\[
    \widehat{\gamma^{}_{\mathrm{F}}} \, =  
    \sum_{k \in \frac{1}{\sqrt{5}\ts} \mathbb{Z} [\tau]}
    I (k) \, \delta^{}_{k}
\]
with intensities $I (k) = \bigl(\frac{\tau}{\sqrt{5}}\, \frac{\sin
  (\pi \tau k')}{\pi \tau k'}\bigr)^{2}$. Here, $\frac{\tau}{\sqrt{5}}
= \frac{\tau + 2}{5}$ is the density of the point set, and $k'$
denotes the algebraic conjugate of $k$, which is defined on the field
$\mathbb{Q} (\sqrt{5}\,)$ by $\sqrt{5} \mapsto - \sqrt{5}$ and acts as
the $\star$-map for the underlying model set description. In
particular, the diffraction is the same for all Dirac combs of the
Fibonacci hull; see \cite[Sec.~9.4.1]{BBG-TAO} and references therein
for details.  An illustration is shown in the upper panel of
Figure~\ref{BBG-fig:Fibo}.

\begin{figure}
\centerline{\includegraphics[width=0.825\textwidth]{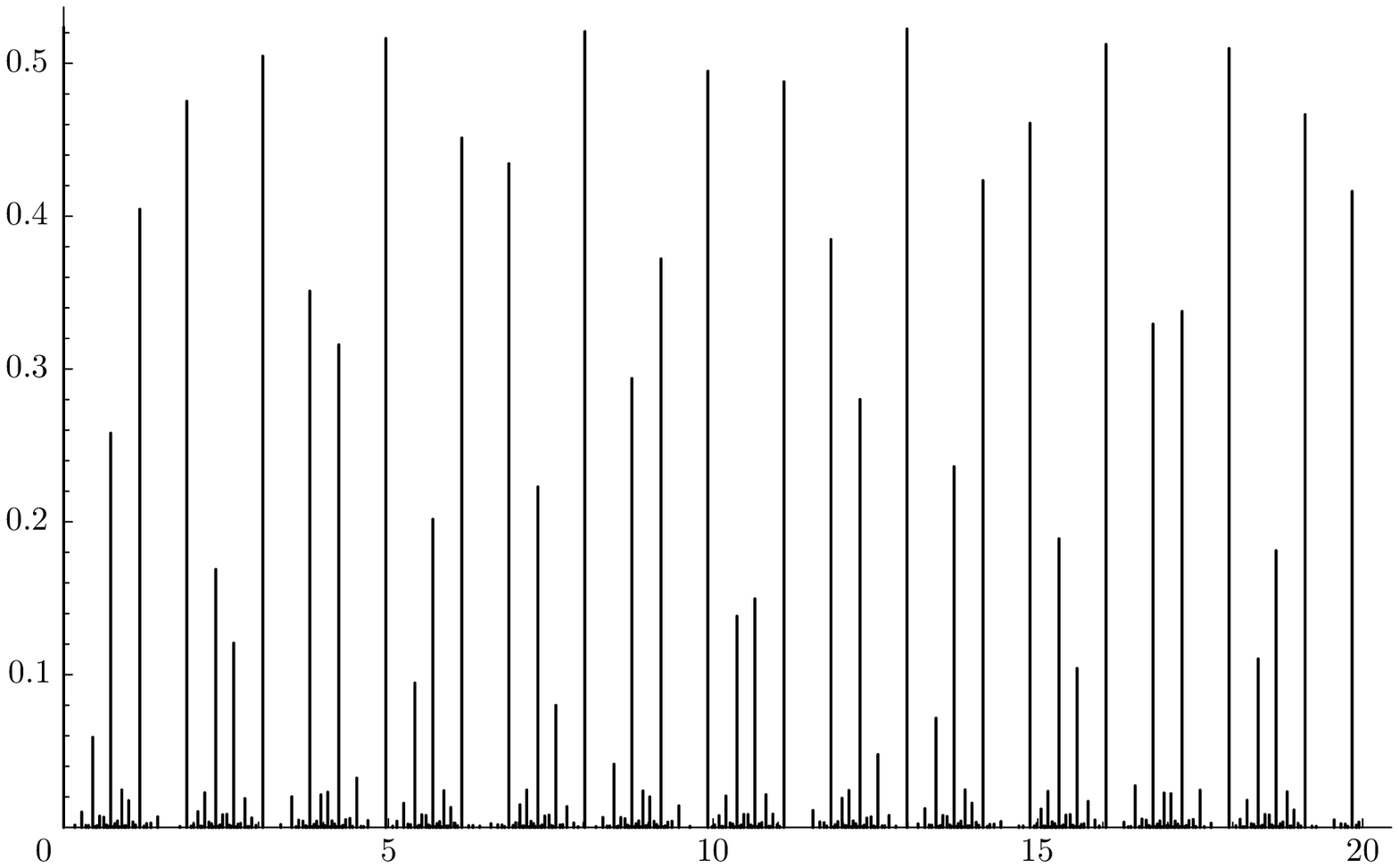}}
\bigskip
\centerline{\includegraphics[width=0.825\textwidth]{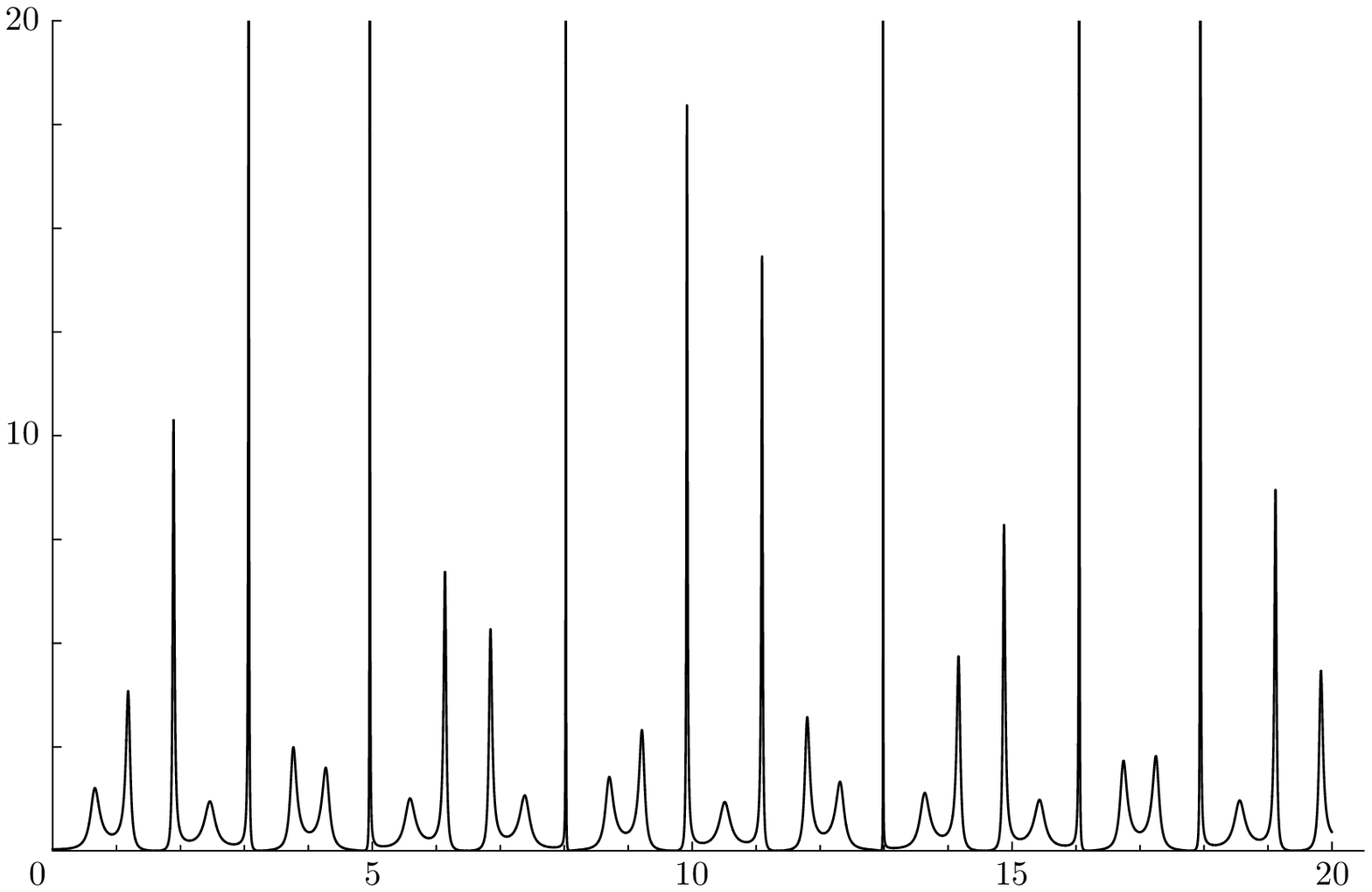}}
\caption{The pure point diffraction measure of the perfect Fibonacci
  chain (upper panel) and the absolutely continuous part of the
  corresponding random tiling (lower panel). Bragg peaks (in the upper
  picture) are shown as lines, where the height is the intensity,
  while the smooth Radon--Nikodym density in the lower picture is
  truncated at a value of $20$ to illustrate the spikyness.  The
  central peak (of intensity $\frac{\tau + 1}{5}$) is omitted in both
  diagrams.}
\label{BBG-fig:Fibo}
\end{figure}

The corresponding random tiling ensemble consists of all tilings of
the real line by the two types of intervals. For a direct comparison,
it makes more sense to only consider those tilings with the same
relative frequency of interval types, which means frequencies $1/\tau$
and $1/\tau^{2}$ for the long and the short interval, respectively.

The diffraction of a typical Dirac comb out of this class was
originally derived in \cite{BBG-BH}, but can also be obtained via an
application of the renewal structure from
Theorem~\ref{BBG-thm:renewal}. This leads to
\[
    \widehat{\gamma^{}_{\mathrm{rt}}} \, = \,
    \frac{\tau + 1}{5} \delta^{}_{0} + h \ts \lambda
\]
with the Radon--Nikodym density
\[
    h(k) \, = \, \frac{\tau + 2}{5} \,
    \frac{(\sin (\pi k/\tau))^{2}}
    {\tau^{2} (\sin (\pi k \tau))^{2} +
    \tau (\sin (\pi k) )^{2} -
    ( \sin (\pi k / \tau))^{2}} \ts .
\]
Except for the trivial Bragg peak at $k=0$, the diffraction measure is
thus absolutely continuous. Still, the resemblance between this
function and the diffraction of the perfect Fibonacci chain is
remarkable, as can be seen from Figure~\ref{BBG-fig:Fibo}.

The situation in dimensions $d\ge 2$ is less favourable from a
mathematical perspective, although one has a rather clear intuition of
what one should expect \cite{BBG-Henley,BBG-RichardDiss}, based on
solid scaling arguments. In dimensions $d\ge 3$, a mixed spectrum with
pure point and absolutely continuous components is conjectured, while
$d=2$ is the critical dimension in the sense that random tilings with
non-crystallographic symmetries should display a singular continuous
component; see \cite[Sec.~11.6.2]{BBG-TAO} for an example.

Unfortunately, only few results have been proved so far. Among them
are a rigorous treatment of planar random tiling ensembles with
crystallographic symmetries (such as the lozenge tiling and several
relatives, see \cite{BBG-BH,BBG-Moritz1,BBG-Moritz2}), a
group-theoretic approach to one of the random tiling hypotheses
\cite{BBG-Richard1,BBG-RichardDiss} and a treatment of dense Dirac
combs with pure point diffraction \cite{BBG-Richard2,BBG-LR07} that is
needed to understand the pure point part of the random tiling
diffraction in dimensions $d\ge 3$. The remaining questions are still
open, though there is little doubt that the original analysis from
\cite{BBG-Henley} is essentially correct.

Let us now leave the realm of explicit examples and turn our attention
to a more general approach of systems with randomness, formulated with
methods from the theory of point processes; compare
\cite{BBG-Gou1,BBG-LM09,BBG-LMpre,BBG-BBM10} for related aproaches and
results.

\section{Stochastic Point Processes and the Palm Measure}

In this section, we take the viewpoint of a general shift-invariant
random measure and relate its realisation-wise diffraction to its
second moment measure.  As such, this section is a complex-valued
extension of \cite[Sec.~5]{BBG-BBM10}.

Let $\mu = \mu^{}_\Re + \mathrm{i} \ts \mu^{}_\Im$ be a locally finite
complex-valued measure on $\mathbb{R}^d$ (which means that
$\mu^{}_\Re$ and $\mu^{}_\Im$ are both locally finite signed
measures). A short calculation reveals that, for $f\in C_{\mathsf{c}}
(\mathbb{R}^{d},\mathbb{C})$ of the form $f = g+ \mathrm{i} h$ with
real-valued $g$ and $h$, the measure $\widetilde{\mu}$ can
consistently be defined via
\[
    \widetilde{\mu}(f) \, := \, \overline{\mu(\widetilde{f}\,)}
    \, = \, \overline{\mu} \ts (f^{}_{-}) \, = \,
   \mu^{}_\Re(f^{}_{-}) -\mathrm{i} \ts \mu^{}_\Im(f^{}_{-}) \ts ,
\]
where $\widetilde{f}(x) = \overline{f^{}_{-}(x)}$ with
$f^{}_{-}(x)=f(-x)$. In particular, note that
\[
   \overline{\mu} \, = \, \mu^{}_\Re - \mathrm{i} \mu^{}_\Im
   \quad \text{and} \quad
   \overline{\mu(f)} \, = \, \overline{\mu}(\ts\overline{f}\ts)
\]
hold as expected. The point here is that, after having dealt 
with the case of real (or signed) measures, the extension to
complex measures is canonical and consistent.

To continue, recall the \emph{polar representation} of a complex
measure from \cite[Ch.~XIII.16]{BBG-D70}; see also
\cite[Prop.~8.3]{BBG-TAO}. Given $\mu$, there is a measurable function
$\alpha^{}_{\mu} \! : \, \mathbb{R}^d \to [0,2\pi)$ such that, for $f
\in C_{\mathsf{c}}(\mathbb{R}^d, \mathbb{C})$, one has
\[
   \int_{\mathbb{R}^d} f(x) \dd\mu(x) \, = 
   \int_{\mathbb{R}^d} f(x)\, \mathrm{e}^{\mathrm{i} \alpha^{}_{\nts\mu}(x)}
   \dd\lvert\mu\rvert (x),
\]
where $\lvert\mu\rvert$ is the total variation measure of $\mu$.  This
means that $\lvert\mu\rvert$ is the smallest non-negative measure such
that $|\mu(A)| \le \lvert\mu\rvert(A)$ for any bounded and measurable
$A$, where $\lvert\mu\rvert \le \lvert\mu_\Re\rvert +
\lvert\mu_\Im\rvert$; compare \cite[Sec.~8.5.1]{BBG-TAO} and
references therein.  \medskip

Let $\mathcal{M}$ denote the $\mathbb{C}$-vector space of all locally
finite, complex-valued measures $\phi$ on $\mathbb{R}^d$, so $\phi \in
\mathcal{M}$ means $|\phi(A)| < \infty$ for any bounded Borel set
$A$. A sequence $(\phi_n)_{n\in\mathbb{N}} \subset \mathcal{M}$
converges vaguely to $\phi$ if $\phi_n(f) \longrightarrow \phi(f)$ as
$n\to\infty$ for all $f \in C_{\mathsf{c}}(\mathbb{R}^d)$.  The space
$\mathcal{M}$ is closed in the topology of vague convergence of
measures (in fact, $\mathcal{M}$ is a Polish space with this
topology). We let $\varSigma_{\mathcal{M}}$ denote the
$\sigma$-algebra of Borel sets of $\mathcal{M}$.  The latter can be
described as the $\sigma$-algebra of subsets of $\mathcal{M}$
generated by the requirement that, for all bounded Borel sets
$A\subset\mathbb{R}^d$, the mapping $\phi\mapsto\phi(A)$ is
measurable.

For each $t\in\mathbb{R}^d$, let $T_t$ denote the translation operator
on $\mathbb{R}^d$, as defined by the mapping $x\mapsto t+x$. Clearly,
$T_t \ts T_s = T_{t+s}$, and the inverse of $T_t$ is given by
$T^{-1}_t = T^{}_{-t}$. For functions $f$ on $\mathbb{R}^d$, the
corresponding translation action is defined via $T_t f = f\circ
T_{-t}$, so that $(T_t f) (x) = f(x-t)$. Similarly, for $\phi \in
\mathcal{M}$, let $T_x \phi := \phi \circ T_{-x}$ be the image measure
under the translation, so that $(T_x \phi)(A) = \phi(T_{-x}(A)) =
\phi(A-x)$ for any measurable subset $A \subset \mathbb{R}^d$, and
$(T_{x} \phi) (f) = \int_{\mathbb{R}^d} f(y) \dd (T_x \phi) (y) =
\int_{\mathbb{R}^d} f(x+z) \dd \phi(z) = \phi (T_{-x} f)$ for
functions.  This means that there is a translation action of
$\mathbb{R}^d$ on $\mathcal{M}$.  Finally, we also have a translation
action on $\mathcal{P} (\mathcal{M})$, the probability measures on
$\mathcal{M}$, via $(T_x Q)(A) = Q(T_{-x}A)$ for any $A \in
\varSigma_{\mathcal{M}}$ and $Q\in\mathcal{P}(\mathcal{M})$. A set $A
\in \varSigma_{\mathcal{M}}$ is called \emph{invariant} (under
translations) if $T_{-x}A = A$ for all $x \in \mathbb{R}^d$.

A (complex-valued) \emph{random measure} $\varPhi$ is a random
variable (defined on some probability space $(\varTheta, \mathcal{F},
\pi)$) with values in $\mathcal{M}$, which formally means that
$\varPhi\! :\, \varTheta \longrightarrow \mathcal{M}$ is an
$(\mathcal{F}\!-\!\varSigma_{\mathcal{M}})$-measurable function. Its
distribution is then $Q = \pi \, \circ \, \varPhi^{-1} \in \mathcal{P}
(\mathcal{M})$, i.e.\ the image measure of $\pi$ under $\varPhi$.  We
will follow the usual practice in probability theory and not make the
underlying probability space explicit (a canonical choice can in
many cases simply be $\varTheta=\mathcal{M}$ and $\varPhi =
\mathrm{Id}_\mathcal{M}$). We will also usually suppress the dependence
of $\varPhi$ on $\theta \in \varTheta$ in the notation.  Integrals
over $\varTheta$ w.r.t.\ the probability measure $\pi$ will be denoted
by $\mathbb{E}$, the expectation value.

$\varPhi$ is called \emph{stationary} if its distribution $Q$
satisfies $T_xQ = Q$ for all $x \in \mathbb{R}^d$.  A stationary
random measure is called \emph{ergodic} if the shift-invariant
$\sigma$-algebra is trivial, which means that any invariant $A$ has
probability $0$ or $1$ (more generally, one requires $Q(A) \in
\{0,1\}$ whenever $Q\bigl((T_{-x}A) \, \triangle \, A\bigr) = 0$ for
all $x \in \mathbb{R}^d$; compare \cite[Def.~10.3.I and
Prop.~10.3.III]{BBG-DVJ}).  \smallskip

In what follows, we generally assume that 
\begin{align}\label{BBG-ass:phiergodic}
  \parbox{0.58\textwidth}{$\varPhi$ is a (possibly) complex-valued,
    stationary and ergodic random measure on $\mathbb{R}^d$,}
\end{align}
which means that there is a decomposition $\varPhi = \varPhi_\Re +
\mathrm{i}\ts \varPhi_\Im$ where both $\varPhi_\Re$ and $\varPhi_\Im$
are signed, real-valued, stationary, ergodic random measures on
$\mathbb{R}^d$.  To verify the last statement note that, since for any
bounded measurable $A \subset \mathbb{R}^d$, $\theta \mapsto
\varPhi(A) \, (=\varPhi(\theta, A)) \in \mathbb{C}$ is measurable,
also $\varPhi_\Re(A)$ and $\varPhi_\Im(A)$ are measurable as functions
of $\theta$. Consider any shift-invariant measurable $B \subset
\mathcal{M}_{\mathrm{real}}$ ($\mathcal{M}_{\mathrm{real}}$ denotes
the locally finite signed measures on $\mathbb{R}^d$), then $\{
\varPhi \mid \varPhi_\Re \in B \}$ is shift invariant and measurable as
well, so $\mathbb{P}(\varPhi_\Re \in B) \in \{0,1\}$, and
analogously for $\varPhi_\Im$.  We further assume that $\varPhi$ is
locally square integrable in the sense that
\begin{align}\label{BBG-ass:2ndmoments}
   \mathbb{E}\left[ \bigl(|\varPhi^{}_\Re|(A)\bigr)^2 + 
   \bigl(|\varPhi^{}_\Im|(A)\bigr)^2 \right]
   < \infty \qquad \text{for all bounded}\; A \subset \mathbb{R}^d,
\end{align}
where $|\varPhi^{}_\Re|$ and $|\varPhi^{}_\Im|$ denote the total
variation measures of $\varPhi^{}_\Re$ and $\varPhi^{}_\Im$,
respectively.  \medskip

In analogy with the real-valued case in \cite[Sec.~5.2]{BBG-BBM10}, we
define $\mu^{(2)}$, the second moment measure of $\varPhi$, via
\begin{equation}\label{BBG-def:mu2}
    \mu^{(2)}(A \times A') \, = \, 
   \mathbb{E}\big[ \varPhi(A)\ts \overline{\varPhi(A')} \ts\ts \big] 
   \qquad \text{for bounded} \;\; A, A' \in \mathcal{B}(\mathbb{R}^d), 
\end{equation}
hence, for $f \in C_{\mathsf{c}}(\mathbb{R}^d \times \mathbb{R}^d, \mathbb{C})$, 
\begin{equation*} 
    \mu^{(2)}(f) \, = \, 
   \mathbb{E}\left[ \int\nolimits_{\mathbb{R}^d} 
   \int\nolimits_{\mathbb{R}^d}  f(x,y)  
    \dd\varPhi(x) \dd\overline{\varPhi}(y) \right] .
\end{equation*}
By the shift invariance of the distribution of $\varPhi$, we have 
\begin{equation*} 
   \int_{\mathbb{R}^d \times \mathbb{R}^d} f(x,y) \dd\mu^{(2)}(x,y) \, = 
   \int_{\mathbb{R}^d \times \mathbb{R}^d} f(x+z,y+z) \dd\mu^{(2)}(x,y) 
\end{equation*}
for all $ z \in \mathbb{R}^d$, and hence we can factor out this
symmetry to obtain the reduced second moment measure
$\mu^{(2)}_\mathrm{red}$. The latter is a locally finite
complex-valued measure that is characterised by
\begin{align} 
\label{BBG-eq:mured:def1}
   \int_{\mathbb{R}^d \times \mathbb{R}^d} f(x,y) \dd\mu^{(2)}(x,y) \, = 
   \int_{\mathbb{R}^d} \int_{\mathbb{R}^d} f(u+v, u) \dd\mu^{(2)}_\mathrm{red}(v) 
   \dd\lambda(u) 
\end{align}
for $f \in C_{\mathsf{c}}(\mathbb{R}^d \times \mathbb{R}^d,
\mathbb{C})$.  By the shift invariance of Lebesgue measure on
$\mathbb{R}^d$, we equivalently have
\begin{align} 
\label{BBG-eq:mured:def1a}
   \int_{\mathbb{R}^d \times \mathbb{R}^d} f(x,y) \dd\mu^{(2)}(x,y) \, = 
   \int_{\mathbb{R}^d} \int_{\mathbb{R}^d} f(u, u-v) 
   \dd\mu^{(2)}_\mathrm{red}(v) \dd\lambda(u). 
\end{align}
To prove the existence of $\mu^{(2)}_\mathrm{red}$, one can decompose
$\mu^{(2)} = \mu^{(2)}_\Re + \mathrm{i} \mu^{(2)}_\Im$ into real and
imaginary parts and then use the well-known real-valued results
(compare \cite[Lemma~10.4.III]{BBG-DVJ}) to obtain
$\mu^{(2)}_\mathrm{red} =\mu^{(2)}_{\Re,\mathrm{red}} + \mathrm{i} \ts
\mu^{(2)}_{\Im,\mathrm{red}}$.

Note that $\mu^{(2)}_\mathrm{red}$ is uniquely defined and is
a positive definite measure, since 
\begin{align} 
\label{BBG-eq:mu2andmured}
  \mu^{(2)} \big( {f} \otimes \overline{g} \big)\, & = 
  \int_{\mathbb{R}^d} \int_{\mathbb{R}^d} f(u+v)\, \overline{g (u)} 
  \dd\mu^{(2)}_\mathrm{red}(v) \dd\lambda(u) \notag \\[1mm]
  & = \int_{\mathbb{R}^d} \int_{\mathbb{R}^d} f(v-w) \,
    \overline{g (-w)} \dd\lambda(w) 
  \dd\mu^{(2)}_\mathrm{red}(v)  \\[1mm]
  & = \int_{\mathbb{R}^d} \big({f} * \widetilde{g} \big)(v) 
  \dd\mu^{(2)}_\mathrm{red}(v) 
  \, = \, \mu^{(2)}_\mathrm{red}\big(f * \widetilde{g}\ts\big), \notag
\end{align}
so that
\begin{align*} 
   \mu^{(2)}_\mathrm{red}\big(f * \widetilde{f} \, \big) 
   \, & = \, \mu^{(2)} \big( {f} \otimes \overline{f} \ts\big) 
   \, = \, \mathbb{E}\Big[ {\mbox{\Large $\int$} f \dd\varPhi \, 
    {\mbox{\Large $\int$}}\ts \overline{f} 
   \dd\overline{\varPhi} }\, \Big] \\[1mm]
  & = \, \mathbb{E}\left[ {{\mbox{\Large $\int$}} {f} \dd\varPhi \, 
  \overline{ {\mbox{\Large $\int$}} {f} \dd\varPhi} \, } \right] 
   \, = \, \mathbb{E}\left[ \lvert \varPhi(f) \rvert^2 
   \right] \ge 0.
\end{align*}

\begin{rem}[{see also \cite[Rem.~13]{BBG-BBM10}}]
  One can alternatively define
\begin{equation*} 
   \mu^{(2,\mathrm{alt})}(A \times A') \, = \, 
   \mathbb{E}\big[\, \overline{\varPhi(A)} \varPhi(A')\ts \big] 
   \qquad \text{for bounded} \;\; A, A' \in 
   \mathcal{B}(\mathbb{R}^d), 
\end{equation*}
and then obtain $\mu^{(2,\mathrm{alt})}_\mathrm{red}$ from this 
as above, via 
\begin{align*}
   \int_{\mathbb{R}^{d}\times\mathbb{R}^{d}} f(x,y) 
   \dd\mu^{(2,\mathrm{alt})}(x,y)\, & = 
   \int_{\mathbb{R}^{d}} \int_{\mathbb{R}^{d}} f(u+v, u) 
   \dd\mu^{(2,\mathrm{alt})}_\mathrm{red}(v) \dd\lambda(u) \\[1mm]
   & = \int_{\mathbb{R}^{d}} \int_{\mathbb{R}^{d}} f(u, u-v) 
  \dd\mu^{(2,\mathrm{alt})}_\mathrm{red}(v) \dd\lambda(u) .
\end{align*}
Then, we have $\mu^{(2,\mathrm{alt})} = \overline{\mu^{(2)}}$ and
$\mu^{(2,\mathrm{alt})}_\mathrm{red} = \overline{\mu^{(2)}_\mathrm{red}}$. 
Since
\begin{equation*}
   \int_{\mathbb{R}^{d}\nts\times\mathbb{R}^{d}} f(y,x)
   \dd\mu^{(2)}(x,y)
  \,  = \int_{\mathbb{R}^{d}\nts\times\mathbb{R}^{d}} f(x,y)
   \dd\overline{\mu^{(2)}}(x,y) 
    \, = \int_{\mathbb{R}^{d}\nts\times\mathbb{R}^{d}} f(x,y)
    \dd\mu^{(2,\mathrm{alt})}(x,y)\ts ,
\end{equation*}
we see that the alternative choice of factoring out the shift
invariance in Eq.~\eqref{BBG-eq:mured:def1}, namely integrating $f(u,
u+v)$ on the right-hand side of this equation, leads to
$\mu^{(2,\mathrm{alt})}_\mathrm{red}$, where
\begin{equation} 
   \int_{\mathbb{R}^d \times \mathbb{R}^d} f(x,y) \dd\mu^{(2)}(x,y) 
   \, = \int_{\mathbb{R}^d} \int_{\mathbb{R}^d} f(u, u+v) 
   \dd\mu^{(2,\mathrm{alt})}_\mathrm{red}(v) \dd\lambda(u) \, .
\end{equation}
We choose the definitions as in Eqs.~\eqref{BBG-def:mu2} and
\eqref{BBG-eq:mured:def1} because these fit well to the formulation of
the limit in Eq.~\eqref{BBG-eq:thmstochprocautocorr} below. Note that,
in the real-valued case, $\mu^{(2)}_\mathrm{red}$ and
$\mu^{(2,\mathrm{alt})}_\mathrm{red}$ agree.
\end{rem}  
\smallskip

The `complex-valued' analogue of \cite[Thm.~5]{BBG-BBM10} now reads as
follows.
\begin{thm}\label{BBG-thm:complrm} 
  Assume that conditions\/ \eqref{BBG-ass:phiergodic} and\/
  \eqref{BBG-ass:2ndmoments} are satisfied, and let\/ $\varPhi_n :=
  \varPhi|^{}_{B_n}$ denote the restriction of\/ $\varPhi$ to the open
  ball of radius\/ $n$ around\/ $0$.  Then, the natural
  autocorrelation of\/ $\varPhi$, which is defined with an averaging
  sequence of nested, centred balls, almost surely exists and
  satisfies
\begin{equation} 
  \label{BBG-eq:thmstochprocautocorr}
   \gamma^{(\varPhi)} \, := \,
   \lim_{n\to\infty} \, \frac{\varPhi_n \! * 
   \widetilde{\varPhi_n}} {\lambda(B_n)} \; = \;
   \lim_{n\to\infty}\, \frac{\varPhi_n \! * 
   \widetilde{\varPhi}} {\lambda(B_n)}
   \; = \;  \mu^{(2)}_\mathrm{red}\ts ,
\end{equation}
where the limit refers to the vague topology. In particular, 
the autocorrelation is non-random. 
\end{thm}

\begin{proof}
  The proof is a suitable `complex-valued interpretation' of the proof
  of \cite[Thm.~5]{BBG-BBM10}.  Fix a continuous function $f\! :\,
  \mathbb{R}^d \longrightarrow \mathbb{C}$ with compact support. We
  have to check that
\begin{equation} 
\label{BBG-claim1}
   \frac{1}{\lambda(B_n)} \bigl( \varPhi_n \! * 
   \widetilde{\varPhi_n}\,\bigr)(f)\; \xrightarrow{\, n\to\infty\,} 
   \; \mu^{(2)}_\mathrm{red} (f)
   \qquad \mbox{(a.s.)}. 
\end{equation}
Since both sides are locally finite (complex-valued) measures, it
actually suffices to check Eq.~\eqref{BBG-claim1} for real-valued $f$.
For $x \in \mathbb{R}^d$, define
\[
   F(x) \, := \int_{\mathbb{R}^d} f(x-y)
  \dd\overline{\varPhi}(y) \, = \int_{\mathbb{R}^d} f(x+y)
  \dd\widetilde{\varPhi}(y) \ts .
\]
Clearly, $F$ inherits stationarity and
ergodicity from $\varPhi$, wherefore $F$ is a (complex-valued) ergodic
random function on $\mathbb{R}^d$ in the sense that shift-invariant
events for $F$ have `trivial' probabilities ($0$ or $1$), and we
obtain
\[
   \mathbb{E} \left[ \int_{A} \big\lvert F(x) \big\rvert \dd 
   \lvert \varPhi \rvert (x) \right]
   \, < \, \infty 
\]
for any bounded and measurable $A \subset \mathbb{R}^d$. 

Define a (complex-valued) additive covariant spatial process $X(A)$ in
the sense of \cite{BBG-NZ79}, indexed by a bounded and measurable $A
\subset \mathbb{R}^d$, via
\begin{equation*} 
   X(A) \, := \int_A F(x) \dd\varPhi(x)\, . 
\end{equation*}
Covariant in this context means that $X$ behaves `naturally' under
translations:\ When $\mathbb{R}^d$ acts on $X$ via $(T_uX)(A) :=
\int_A F(x) \dd (T_u\varPhi)(x)$, for $u\in\mathbb{R}^{d}$, then 
$(T_uX)(A+u)=X(A)$.

Decomposing $X$ into its real and imaginary parts (by decomposing $F$
and $\varPhi$ and suitably grouping terms) we can apply
\cite[Cor.~4.9]{BBG-NZ79} to obtain a.s.\
\begin{align*} 
  \lim_{n\to\infty} & \frac{1}{\lambda(B_n)} 
     \bigl( \varPhi_n \nts * \widetilde{\varPhi} \ts\ts \bigr) (f)
  \, =  \lim_{n\to\infty} \frac{1}{\lambda(B_n)} 
     \int_{B_n}\!\! F(x) \dd\varPhi(x) 
  \, = \lim_{n\to\infty} \frac{X(B_{n})}{\lambda(B_n)}  \\
  & = \, \mathbb{E} \left[ \frac{ X({B_1})}{\lambda(B_1)} \right] 
  \, = \, \frac{1}{\lambda(B_1)} \,
   \mathbb{E} \left[ \int_{B^{}_1} \int_{\mathbb{R}^d} f(x-y) 
   \dd\overline{\varPhi}(y) \dd\varPhi(x)\right]  \\
  & =\, \frac{1}{\lambda(B_1)} \int_{\mathbb{R}^d \times\ts \mathbb{R}^d} 
       \mathbf{1}_{B^{}_1}(x) \, f(x-y) \dd\mu^{(2)}(x,y) \\[1mm]
  & =\, \frac{1}{\lambda(B_1)} 
        \int_{\mathbb{R}^d} \int_{\mathbb{R}^d}\mathbf{1}_{B^{}_1}(x) \, f(z) 
        \dd\mu^{(2)}_\mathrm{red}(z) \dd \lambda(x) 
   \, =   \int_{\mathbb{R}^d} f \dd\mu^{(2)}_\mathrm{red}\ts .
\end{align*}
The difference between $\varPhi_n * \widetilde{\varPhi}$ and
$\varPhi_n * \widetilde{\varPhi_n}$ is a (random) `boundary term' that
almost surely vanishes in the limit as $n\to\infty$. To prove this
formally, decompose $\varPhi = \varPhi_\Re + \mathrm{i} \ts
\varPhi_\Im$, $\widetilde{\varPhi} = \widetilde{\varPhi_\Re} -
\mathrm{i} \ts \widetilde{\varPhi_\Im}$ and then argue as in the proof
of \cite[Thm.~3]{BBG-BBM10} for each of the four terms appearing in
$\varPhi_n * \bigl( \widetilde{\varPhi} -
\widetilde{\varPhi_n}\bigr)$.
\end{proof}

\begin{rem}
  Theorem~\ref{BBG-thm:complrm} allows to reformulate Theorem~4 and
  Corollary~1 from \cite{BBG-BBM10} for complex-valued clusters as
  follows. If $\varPhi$ is a stationary ergodic point process, i.e.\
  $\varPhi$ is a random sum of Dirac measures, with distribution $P$
  satisfying Eq.~\eqref{BBG-ass:2ndmoments}, and if we replace each point
  independently by a random complex-valued measure with distribution
  $Q$, then the formulas describing the autocorrelation and the
  diffraction of the resulting cluster process given in \cite[Thm.~4
  and Cor.~1]{BBG-BBM10} continue to hold.
\end{rem}

Let us also mention that, by specialising $\varPhi$ to a renewal
process, Theorem~\ref{BBG-thm:complrm} allows to recover
Eq.~\eqref{BBG-eq:auto-1} and, in particular,
Theorem~\ref{BBG-thm:renewal} from this more general perspective; see
\cite{BBG-BBM10} for further details, and how this can be used to
formulate the renewal process also for more general `dropping'
distributions.

\subsection{A `Palm-type Distribution' for Complex-valued 
Random Measures} 

In the case of a positive random measure $\varPhi$,
Eq.~\eqref{BBG-eq:thmstochprocautocorr} can be interpreted via the Palm
distribution $P_0$ of the law of $\varPhi$, which is a probability
measure on locally finite measures (intuitively, the law of $\varPhi$
viewed relative to a typical point of its support) via
\begin{equation} 
  \label{BBG-eq:mu2redandpalmintensity}
  \mu^{(2)}_\mathrm{red} \, = \, \rho \, I_{P_0}
\end{equation}
where $\rho > 0$ is the intensity and $I_{P_0}$ the first moment
measure of $P_0$; compare \cite[Sec.~5.2]{BBG-BBM10}.  This
interpretation breaks down in general in the signed or complex-valued
case because $\mu^{(2)}_\mathrm{red}$ will not be a positive
measure. One way to extend this line of thought is to re-interpret the
Palm distributions in a way suited for complex-valued random measures
as follows.

Recalling the structure of the polar decomposition, the random measure
$\varPhi$ can equivalently be described via $(\lvert \varPhi \rvert,
\varPhi_{\rm ph})$, where $\lvert \varPhi \rvert$ is the total
variation measure and the mapping $\varPhi_{\rm ph} \!:\, \mathbb{R}^d
\longrightarrow [0,2\pi)$ the `phase function':
\begin{equation} 
  \int_{\mathbb{R}^d} f(x)  \dd\varPhi(x) \, = 
  \int_{\mathbb{R}^d} f(x) \, \mathrm{e}^{\mathrm{i} \varPhi_{\rm ph}(x)}  
   \dd\lvert \varPhi \rvert(x)\ts . 
\end{equation}
Note that $\varPhi \mapsto (\lvert \varPhi \rvert, \varPhi_{\rm ph})$
is measurable, so $(\lvert \varPhi \rvert, \varPhi_{\rm ph})$ is in
fact a random variable.  Define a positive $\sigma$-finite measure
$\mathcal{C}$ on $\mathbb{R}^d \times \mathcal{M}$ (this is the
equivalent of the so-called Campbell measure for the complex-valued
context and agrees with the usual Campbell measure if $\varPhi$ is a
positive random measure) via
\[ 
   \int_{\mathbb{R}^d \times \mathcal{M}} g(x, \varphi)
    \dd \mathcal{C}(x,\varphi) \, := \, 
    \mathbb{E}\Big[ \int_{\mathbb{R}^d} 
    g\big(x, \mathrm{e}^{-\mathrm{i}\varPhi_{\rm ph}(x)}
    \varPhi\big) \dd\lvert \varPhi \rvert(x) \Big] , 
\]
whenever the right-hand side is defined (which will for instance
always be the case when $g$ is measurable and non-negative). By the
shift invariance of $\varPhi$, and hence that of $\lvert \varPhi
\rvert$, the projection of $\mathcal{C}$ to $\mathbb{R}^d$ is $\rho$
times Lebesgue measure (with $\rho \in [0,\infty)$ being the intensity
of $\lvert \varPhi \rvert$), hence there is a family of probability
measures $P_x$ on $\mathcal{M}$, with $P_x \in
\mathcal{P}(\mathcal{M})$ for all $x \in \mathbb{R}^d$, so that we can
disintegrate (compare \cite[Thm.~15.3.3]{BBG-Ka})
\begin{equation} 
   \int g \dd\mathcal{C} \, =  
   \int_{\mathbb{R}^d} \int_{\mathcal{P}(\mathcal{M})} g(x, \varphi) \, 
   \dd P_x(\varphi) \, \rho\dd\lambda(x)\ts . 
\end{equation} 

\begin{defn} 
  We call the elements of the family $\bigl\{P_x \mid x\in
  \mathbb{R}^d\bigr\}$ the \emph{Palm distributions} in the
  complex-valued case.
\end{defn}
Let, for $A \subset \mathbb{R}^d $ bounded and measurable,
\[
   I_{P_x} (A) \, := \int_{\mathcal{M}} \varphi(A) \dd P_x(\varphi)
 \]
 be the expectation (or first moment) measure of $P_x$.  By shift
 invariance, we have $P_x = T_x P^{}_{0}$, $x \in \mathbb{R}^d$, and
 hence $I_{P_x} = T_x I_{P^{}_{0}}$.  The connection between the
 (reduced) second moment measure and the Palm distribution carries
 over to the complex-valued case as follows.
\begin{prop} 
  For the extended definition of the Palm distribution, one has
\[ 
     \mu^{(2)}_{\mathrm{red}} \, =\, \rho \, I^{}_{\nts P^{}_{0}} \ts, 
\] 
so Eq.~\eqref{BBG-eq:mu2redandpalmintensity} also holds in this case.
\end{prop}

\begin{proof}[Sketch of Proof.]
Consider $g(x, \varphi) =
\mathbf{1}_{A'}(x) \ts\varphi(A)$ with $\mathbf{1}$ denoting the
characteristic function and with $A, A' \subset \mathbb{R}^d$ bounded
and measurable.  Then,
\begin{align} 
\label{BBG-eq:complcampbell1}
   \int g \dd\mathcal{C} & \, = \,
   \mathbb{E}\Big[ \int_{\mathbb{R}^d}  
   \mathbf{1}_{A'}(x) \, \mathrm{e}^{-\mathrm{i}\varPhi_{\rm ph}(x)} 
   \,\varPhi(A) \dd\lvert \varPhi \rvert(x) \Big]  \\[1mm]
  & = \, \mathbb{E}\Big[ \varPhi(A) \int_{\mathbb{R}^d}  
   \mathbf{1}_{A'}(x)\, \mathrm{e}^{-\mathrm{i}\varPhi_{\rm ph}(x)}  
   \dd\lvert \varPhi \rvert(x) \Big] 
   \, = \, \mathbb{E}\big[ \varPhi(A)\ts \overline{\varPhi(A')} 
  \, \big] \notag
\end{align}
by definition, whereas the disintegration formula yields 
\begin{align} \label{BBG-eq:complcampbell2}
  \int g \dd\mathcal{C} & \, = 
  \int_{\mathbb{R}^d} \int_{\mathcal{M}} \mathbf{1}_{A'}(x) \varphi(A) \, 
  \dd P_x(\varphi)\, \rho\dd\lambda(x) 
   \, = \, \rho \int_{A'} I_{P_x}(A) \dd\lambda(x) \notag \\[1mm]
  & = \, \rho \int_{A'} I_{P_0}(A-x) \dd\lambda(x) 
   \, = \, \rho \int_{\mathbb{R}^d} \int_{\mathbb{R}^d} 
   \mathbf{1}_{A'}(x) \mathbf{1}_{A-x}(y)  
   \dd I_{P_0}(y) \dd\lambda(x) \notag \\[1mm]
   & = \, \rho \int_{\mathbb{R}^d} \int_{\mathbb{R}^d} \mathbf{1}_{A'}(x) 
   \mathbf{1}_{A}(y+x)  \dd\lambda(x) \dd I_{\nts P^{}_{0}}(y)  \\[1mm]
  & = \, \rho \int_{\mathbb{R}^d} \int_{\mathbb{R}^d} \mathbf{1}_{A'}(-x) 
   \mathbf{1}_{A}(y-x)  \dd\lambda(x) \dd I_{\nts P^{}_{0}}(y) 
   \, = \, \rho\, I_{P_0}\big( \mathbf{1}_{A} * 
   \widetilde{\mathbf{1}_{A'}} \big). \notag
\end{align}
Comparing
Eqs.~\eqref{BBG-eq:complcampbell1}--\eqref{BBG-eq:complcampbell2} with
Eq.~\eqref{BBG-eq:mu2andmured} yields the claim.  
\end{proof}

If $\varPhi$ is ergodic, the viewpoint that $P_0$ describes the
configuration relative to a point in the support drawn according to
$\varPhi$ is corroborated by
\[ 
   \frac{1}{\lambda(B_n)} \int_{B_n} 
   g\big( \mathrm{e}^{-\mathrm{i} \varPhi_{\rm ph}(x)}T_{-x}\varPhi \big) 
   \dd \lvert \varPhi \rvert(x) 
   \; \xrightarrow{\, n\to\infty\,} \,
   \int_{\mathcal{M}} g(\varphi)  \dd P_0(\varphi) \quad 
    \text{(a.s.)} 
\]
for any bounded measurable $g\! :\, \mathcal{M} \longrightarrow
\mathbb{R}$.
\medskip

The viewpoint of (possibly complex-valued) ergodic random measures for
diffraction is a useful one since it provides a connection to the
large literature on random measures and on stochastic geometry; see
\cite{BBG-DVJ,BBG-Karr,BBG-Kuel1,BBG-Kuel2,BBG-Kai,BBG-BBM10} and
references therein, as well as \cite{BBG-TBB} for a recent
generalisation that can also be considered from the diffraction point
of view.  However, our approach also shows a limitation that one
encounters when trying to infer properties of a random configuration
of scatterers from its kinematic diffraction: As is evident from
Eq.~\eqref{BBG-eq:thmstochprocautocorr} in
Theorem~\ref{BBG-thm:complrm}, the only `datum' from a random
$\varPhi$ visible in its autocorrelation, and hence also in the
corresponding diffraction, is the second moment measure. It is well
known that second moments are generally insufficient to determine the
distribution of $\varPhi$ unless further structural properties are
known. This inverse problem is known as the homometry problem in
crystallography and the inference problem in the theory of stochastic
processes.

\section{Outlook}

Our exposition provides a snapshot of the present knowledge about
systems with continuous diffraction components; see \cite[Chs.~10 and
11]{BBG-TAO} as well as \cite{BBG-BBM10,BBG-BKM} for additional
examples, and \cite{BBG-BLvE,BBG-BKM} for connections with the
dynamical spectrum. Nevertheless, as is apparent from a comparison
with the pure point diffraction case
\cite{BBG-Crelle,BBG-TAO,BBG-LMpre,BBG-TB,BBG-T}, the status of
general results is lagging behind. Even for many important examples,
some of the most obvious questions are still open from a mathematical
point of view. In particular, this is so for random tiling ensembles
in dimensions $d\ge 2$, or for equilibrium systems just beyond the
complexity of the (planar) Ising model.

Apart from the systems considered here, an interesting class is
provided by random substitution and inflation systems, as introduced
in \cite{BBG-GL}. The randomness present here is compatible with the
long-range order of Meyer sets with entropy
\cite{BBG-BM12,BBG-Moll,BBG-Moll2}, which means that one obtains
interesting mixtures of pure point and absolutely continuous
diffraction measures. Though this direction has not attracted much
attention so far, it is both tractable and practically relevant.

{}From a more general perspective, one lacks some kind of analogue to
the key theorems in pure point diffraction (such as the Poisson
summation formula or the Halmos--von Neumann theorem). While there is
at least the theory of Riesz products \cite{BBG-Z,BBG-Q} for
self-similar systems with singular spectra, a general approach to
stochastic systems is only at its beginning. Methods from point
process theory \cite{BBG-DVJ}, such as the Palm measure and its
connection to the autocorrelation (via its intensity measure), look
promising, but have not produced many concrete results so far. The
latter, however, are needed to make some progress with the complicated
inverse problem for such systems. Though there is substantial
knowledge from the inference approach \cite{BBG-Karr}, it is not clear
at present how this can be used, and how reasonable restrictions could
be included.

\subsection*{Acknowledgment}
It is our pleasure to thank Aernout van Enter, Holger K\"{o}sters,
Daniel Lenz, Robert Moody and Tom Ward for interesting discussions and
for their cooperation on some of the papers that form the basis for
this review.  This work was supported by the German Research Council
(DFG), within the CRC 701.

\end{document}